\documentclass[a4paper,11pt]{article}
\usepackage[top=30truemm,bottom=30truemm,left=25truemm,right=25truemm]{geometry}

\usepackage{amsthm}
\usepackage{mathtools,mathrsfs,amsfonts}
\usepackage{comment}

\usepackage{graphicx}
\usepackage{hyperref}

\theoremstyle{plain}
\newtheorem{theorem}{Theorem}[section]

\newtheorem{lemma}[theorem]{Lemma}

\newtheorem{proposition}[theorem]{Proposition}

\begin{document}

\title{Discretization of Camassa-Holm peakon equation using orthogonal polynomials and matrix $LR$ transformations}

\author{R. Watanabe\footnote{Graduate School of Informatics, Kyoto University, Kyoto, 606-8501, Japan}, 
\and 
M. Iwasaki\footnote{Faculty of Life and Environmental Sciences, Kyoto Prefectural University, Kyoto, 606-0823, Japan} 
\and and \and S. Tsujimoto$^*$}

\maketitle


\begin{abstract}
Discrete integrable systems are closely related to orthogonal polynomials and isospectral matrix transformations. In this paper, we use these relationships to propose a nonautonomous time-discretization of the Camassa-Holm (CH) peakon equation, which describes the motion of peakon waves, which are soliton waves with sharp peaks. 
We then validate our time-discretization, and clarify its asymptotic behavior as the discrete-time goes to infinity. We present numerical examples to demonstrate that the proposed discrete equation captures peakon wave motions. 
\end{abstract}

%
\vspace{2pc}
\noindent{Keywords}: Camassa-Holm equation, peakon, orthogonal polynomials, $LR$ transformations, integrable systems

\vspace{1pc}
\noindent{E-mail}: watanabe.ryoto.37p@st.kyoto-u.ac.jp
%
%
%
%
%
%
\section{Introduction}
\label{sec:1}
%
Some integrable systems are being found to have intriguing relationships with orthogonal polynomials and isospectral matrix transformations. The Toda equation~\cite{Toda_1967}, which describes the dynamics of nonlinear springs with an exponential potential, is one such integrable system. Time evolutions of the Toda equation correspond to isospectral transformations of tridiagonal matrices~\cite{Symes_1982}. 
Eigenfunctions of the Toda equation, which do not change except for scalar multiplication when subjected to a linear operator corresponding to the time evolution by the Toda equation, are orthogonal polynomials~\cite{Chihara_1978}.
A time-discretization of the Toda (dToda) equation was first designed by using a skillful bilinear form technique~\cite{Hirota_1993}. Discrete-time evolutions by the dToda equation generate tridiagonal $LR$ isospectral transformations that decompose tridiagonal matrices into products of bidiagonal matrices and then reverse the order of the products. 
Thus, the dToda equation is simply the recursion formula of the quotient-difference (qd) algorithm for computing tridiagonal eigenvalues~\cite{Hirota_1993}. 
The dToda equation can also be derived by focusing on the compatibility condition of discrete-time evolutions of orthogonal polynomials by the Christoffel-Darboux transformation and their inverse evolutions by the Geronimus transformation~\cite{Spiridonov_1995,Zhedanov_1997}. 

An integrable system closely related to the Toda equation is the Lotka-Volterra (LV) system~\cite{Kac_1975,Manakov_1975}, which describes the simplest predator-prey relationships. 
Similar to the Toda case, the LV system is discretized with respect to the time variable using bilinear forms~\cite{Tsujimoto_1993} and symmetric orthogonal polynomials~\cite{Spiridonov_1997}. 
Since discrete-time evolutions by the discrete LV (dLV) system can generate symmetric tridiagonal $LR$ transformations~\cite{Iwasaki_2002,Tsujimoto_2002}, the dLV system can be applied to compute singular values of bidiagonal matrices, which is equivalent to computing symmetric tridiagonal eigenvalues. 
To model more complicated predator-prey relationships, the dLV system is extended to discrete hungry LV (dhLV) systems~\cite{Suris_1996a,Tsujimoto_1993}. 
Moreover, similar to the dToda-dLV relationship, the discrete hungry Toda (dhToda) equations can be developed as extensions of the dToda equation by transforming dhLV systems~\cite{Fukuda_2013a,Tokihiro_1999}. 
These discrete hungry integrable systems are related to $LR$ transformations of Hessenberg matrices~\cite{Fukuda_2013b}. 
The relativistic discrete Toda equation~\cite{Suris_1996b} and the elementary Toda orbits~\cite{Kobayashi_2021} are both special cases of the dhToda equation~\cite{Yamamoto_2022} derived by observing spectral transformations of Laurent biorthogonal and $\epsilon$-Laurent biorthogonal polynomials, respectively. 

The Camassa-Holm (CH) equation~\cite{Camassa_1993,Camassa_1994} is an integrable shallow water equation defined by:
\begin{equation}\label{eqn:CH}
u_t-u_{xxt}+3uu_x-2u_xu_{xx}-uu_{xxx}=0,
\end{equation}
where $u\coloneqq u(x,t)$ denotes the height of the fluid surface in direction $x$ at continuous-time $t$, and the subscripts signify the partial derivatives with respect to these terms. 
Remarkably, the CH equation \eqref{eqn:CH} can describe soliton waves with sharp peaks, known as peakon waves. 
Assuming that the CH equation \eqref{eqn:CH} has a peakon solution, we can obtain a time evolution equation with respect to the amplitudes and positions of peakon waves. 
Beals et al.~\cite{Beals_2000,Beals_2001} related the time evolution equation to a discrete string problem, and solved it by using the inverse spectral method~\cite{Moser_1975}, while elucidating the relationship to the Toda equation. Ragnisco et al.~\cite{Ragnisco_1996} showed that the CH peakon equation and the Toda equation satisfy the common Lie-Poisson structure and presented a time-discretization of the CH peakon equation. 
See also Refs.~\cite{Chang_2018_BKP},~\cite{Chang_2018_CKP} and \cite{Chang_2020} concerning relationships of the CH equation to other peakon equations and the corresponding Toda equations. 

In this paper, we propose a new time-discretization of the CH peakon equation by considering extensions of studies on the dToda equation. The key aim of this paper is to extend spectral transformations of orthogonal polynomials and $LR$ transformations of matrices related to the dToda equation. 
Our discrete CH peakon equation differs from the Ragnisco equation~\cite{Ragnisco_1996} in that it is a nonautonomous equation involving arbitrary parameters. In fact, its autonomous equation is the Ragnisco equation. 
Moreover, by using our discrete CH peakon equation, discrete-time evolutions with suitable discretization parameters enable us to observe motions of the peakon waves on the corresponding discrete time scale. We numerically demonstrate this perspective in the third example of Section \ref{sec:5}.

The remainder of this paper is organized as follows. 
In Section \ref{sec:2}, we describe the CH peakon equation and the related isospectral transformations of tridiagonal matrices. 
In Section \ref{sec:3}, we focus on spectral transformations of orthogonal polynomials, and then derive discrete-time evolution equations from recursion formulas concerning  orthogonal polynomials. 
In Section \ref{sec:4}, we relate the discrete-time evolution equations to $LR$-like transformations, and then design a new discrete CH peakon equation. 
In Section \ref{sec:5}, we show the determinantal solution to the proposal discrete equation, and clarify the asymptotic behavior as discrete-time goes to infinity 
by expanding the determinantal solution. We also numerically confirm that the proposed discrete equation captures the motion of peakon waves. 
Finally, in Section \ref{sec:6}, we give concluding remarks. 
%
%
\section{Camassa-Holm peakon equation}
\label{sec:2}
%
In this section, we review the part of Refs.~\cite{Ragnisco_1996} that discusses the relationship of time evolutions in the CH peakon equation to isospectral transformations of tridiagonal matrices. 
This relationship is the starting point in deriving a time-discretization of the CH peakon equation in later sections. 

From Camassa et al.~\cite{Camassa_1993}, the $N$-peakon solution to the CH equation \eqref{eqn:CH} is given by:
\begin{equation}\label{eqn:CH_solution}
u(x,t)=\sum_{i=1}^Np_i(t)e^{-\vert x-x_i(t) \vert}, 
\end{equation}
where $p_i(t)$ and $x_i(t)$ are the amplitude and position of the $i$th peakon, respectively. 
Without loss of generality, we can index peakons as $1,2,\dots,N$ in order of distance from the origin at $t=0$. In other words, we may assume that $x_{1}(0)>x_{2}(0)>\cdots>x_{N}(0)$. 
Moreover, we assume that $p_1(0)>0,p_2(0)>0,\ldots,p_N(0)>0$. 
According to Holden et al.~\cite{Holden_2006}, $p_1(t)>0,p_2(t)>0,\ldots,p_N(t)>0$ and $x_{1}(t)>x_{2}(t)>\cdots>x_{N}(t)$ at any $t$ under the assumptions, and $p_i(t)$ and $x_i(t)$ satisfy the following peakon equation \eqref{eqn:peakon}:
\begin{equation}\label{eqn:peakon}
\left\{\begin{aligned}
&\frac{dp_i(t)}{dt}=-p_i(t)\sum_{j=1}^{i-1}p_j(t)e^{x_i(t)-x_j(t)}
+p_i(t)\sum_{j=i+1}^Np_j(t)e^{x_j(t)-x_i(t)},\quad i=1,2,\ldots,N,\\
&\frac{dx_i(t)}{dt}=\sum_{j=1}^{i-1} p_j(t) e^{x_i(t)-x_j(t)}
+p_i(t)+\sum_{j=i+1}^{N} p_j(t) e^{x_j(t)-x_i(t)},\quad i=1,2,\ldots,N.
\end{aligned}\right.
\end{equation}
In this paper, we refer to \eqref{eqn:peakon} as the peakon equation shortly. 

We now introduce an $N$-by-$N$ symmetric matrix ${\cal A}(t)$ involving the peakon variables $p_i(t)$ and $x_i(t)$:
\[
{\cal A}(t) \coloneqq 
\left( \begin{array}{cccc}
p_{1}(t) & \sqrt{p_{1}(t)p_{2}(t)} e^{\frac{x_{2}(t)-x_{1}(t)}{2}} & \cdots & \sqrt{p_{1}(t)p_{N}(t)} e^{\frac{x_{N}(t)-x_{1}(t)}{2}} \\
\sqrt{p_{1}(t)p_{2}(t)} e^{\frac{x_{2}(t)-x_{1}(t)}{2}} & p_{2}(t) & \cdots & \sqrt{p_{2}(t)p_{N}(t)} e^{\frac{x_{N}(t)-x_{2}(t)}{2}} \\
\vdots & \vdots & \ddots & \vdots \\
\sqrt{p_{1}(t)p_{N}(t)} e^{\frac{x_{N}(t)-x_{1}}{2}(t)} & \sqrt{p_{2}(t)p_{N}(t)} e^{\frac{x_{N}(t)-x_{2}(t)}{2}} & \cdots & p_{N}(t)
\end{array} \right).
\]
Thus, we obtain a matrix representation of the peakon equation \eqref{eqn:peakon}. 
\begin{proposition}[cf.~\cite{Ragnisco_1996}]\label{prop:Lax_representation_peakon}
The peakon equation \eqref{eqn:peakon} can be represented in matrix form as:
\begin{equation}\label{eqn:Lax_representation_peakon}
\frac{d{\cal A}(t)}{dt}={\cal A}(t)\Pi({\cal A}(t))-\Pi({\cal A}(t)){\cal A}(t),
\end{equation}
where $\Pi({\cal A}(t))$ denotes an antisymmetric matrix determined using the strictly lower triangular part $({\cal A}(t))_<$ and the strictly upper triangular part $({\cal A}(t))_>$ of the matrix ${\cal A}(t)$ as follows:
\[
\Pi({\cal A}(t)) \coloneqq 
\frac{1}{2}({\cal A}(t))_<- \frac{1}{2}({\cal A}(t))_>
\]
\end{proposition}
\noindent
Here, $\sigma_k$ and $\Psi_k(t)$ denote an eigenvalue of ${\cal A}(t)$ and the corresponding eigenvector, respectively. 
Then, the matrix representation \eqref{eqn:Lax_representation_peakon} is equivalent to the compatibility condition of the linear equations:
\begin{equation}\label{eqn:Lax_pair}
\displaystyle
\left\{\begin{aligned}
&{\cal A}(t)\Psi_i(t)=\sigma_i\Psi_i(t),\quad i=1,2,\ldots,N,\\
&\frac{d\Psi_i(t)}{dt}=-\Pi({\cal A}(t))\Psi_i(t),\quad i=1,2,\ldots,N.
\end{aligned}\right.
\end{equation}
Thus, time evolutions in the peakon equation \eqref{eqn:peakon} generate isospectral transformations of the symmetric tridiagonal ${\cal A}(t)$. 

We emphasize that ${\cal A}(t)$ is a symmetric semiseparable matrix~\cite{Vandebril_2008}. 
An $N$-by-$N$ matrix $S$ is called symmetric semiseparable if the $(i,j)$ entries of $S$, denoted $s_{i,j}$, can be expressed in product form as:
\[
s_{i,j}=
\left\{\begin{aligned}
u_iv_j,  i=1,2,\ldots,N,\quad j=1,2,\ldots,i,\\
u_jv_i,  i=1,2,\ldots,N,\quad j=i,i+1,\ldots,N.
\end{aligned}\right.
\]
Setting $u_1=\sqrt{p_1(t)},u_2=\sqrt{p_2(t)}e^{\frac{x_2(t)-x_1(t)}{2}},\ldots,u_N=\sqrt{p_N(t)}e^{\frac{x_N(t)-x_1(t)}{2}}$ and $v_1=\sqrt{p_1(t)},v_2=\sqrt{p_2(t)}e^{\frac{x_1(t)-x_2(t)}{2}},\ldots,v_N=\sqrt{p_N(t)}e^{\frac{x_1(t)-x_N(t)}{2}}$, we can easily check that $S={\cal A}(t)$. 

The following proposition describes the inverse of a symmetric semiseparable matrix.
\begin{proposition}[cf.~\cite{Vandebril_2008}]\label{prop:singlepair_tridiagonal}
Let $S$ be an invertible symmetric semiseparable matrix with $u_1\ne0,u_2\ne0,\ldots,u_N\ne0$ and $v_1\ne0,v_2\ne0,\ldots,v_N\ne0$. Then the inverse matrix $S^{-1}$ is a symmetric tridiagonal matrix. 
\end{proposition}
\noindent
We focus on the case where $\sigma_1\ne0,\sigma_2\ne0,\ldots,\sigma_N\ne0$. Obviously, ${\cal A}(t)$ is invertible. 
From Proposition \ref{prop:singlepair_tridiagonal}, we also see that $A(t)\coloneqq({\cal A}(t))^{-1}$ is symmetric tridiagonal with the form:
\[A(t)=\left( \begin{array}{cccc}
a_{1}(t) & b_{1}(t) & \\ b_{1}(t) & a_{2}(t) & \ddots \\
 & \ddots & \ddots & b_{N-1}(t) \\ & & b_{N-1}(t) & a_N(t)
\end{array} \right).\]
We can actually express $a_i(t)$ and $b_i(t)$ using the peakon variables $p_i(t)$ and $x_i(t)$ as: 
\begin{equation}\label{eqn:peakon_Toda_variables}
\displaystyle
\left\{\begin{aligned}
&a_i(t)\coloneqq \frac{1}{p_i(t)} 
\frac{ 1-e^{x_{i+1}(t)-x_{i-1}(t)} }{ (1-e^{x_i(t)-x_{i-1}(t)})(1-e^{x_{i+1}(t)-x_i(t)}) },\quad i=1,2,\ldots,N,\\
&b_i(t)\coloneqq-\frac{ 1 }{ \sqrt{p_i(t)p_{i+1}(t)} }
\frac{ e^{ \frac{ x_{i+1}(t)-x_i(t) }{2} } }{ 1-e^{x_{i+1}(t)-x_i(t)} },\quad i=1,2,\ldots,N-1.
\end{aligned}\right.
\end{equation}
Moreover, as ${\cal A}(t)=(A(t))^{-1}$, we can rewrite \eqref{eqn:Lax_representation_peakon} in terms of a tridiagonal matrix.
\begin{theorem}\label{thm:Lax_representation_peakon_Toda}
A time evolution of tridiagonal $A(t)$ is given as:
\begin{equation}\label{eqn:Lax_representation_peakon_Toda}
\frac{dA(t)}{dt}=A(t)\Pi((A(t))^{-1})-\Pi((A(t))^{-1})A(t).
\end{equation}
\end{theorem}
\noindent
From Refs~\cite{Flaschka_1974_I} and \cite{Flaschka_1974_II} that the nonperiodic finite Toda equation satisfies $d\check{A}(t)/dt=\check{A}(t)\Pi(\check{A}(t))-\Pi(\check{A}(t))\check{A}(t)$, where $\tilde{A}(t)$ is symmetric tridiagonal. 
Thus, we find a relationship of the peakon equation \eqref{eqn:peakon} to the Toda equation
from the perspective of matrix representations. 
%
%
\section{Spectral transformations of orthogonal polynomials}
\label{sec:3}
%
In this section, we briefly explain some properties concerning spectral transformations of orthogonal polynomials. 

We consider the orthogonal polynomials $\phi_i(z)$ that satisfy the following three-term relationship:
\begin{equation*}
\begin{aligned}
&\phi_{-1}(z)\coloneqq0,\quad\phi_0(z)\coloneqq1,\\
&\phi_i(z)\coloneqq(z-a_i)\phi_{i-1}(z)-b_{i-1}^2\phi_{i-2}(z),\quad i=1,2,\ldots,N,
\end{aligned}
\end{equation*}
where $a_1,a_2,\ldots,a_N$ and $b_1\ne0,b_2\ne0,\ldots,b_{N-1}\ne0$ are constants. We can easily check that $\phi_1(z),\phi_2(z),\ldots,\phi_N(z)$ are $1$st-, $2$nd$-,\ldots,N$th-order monic polynomials, respectively. According to Favard's theorem~\cite{Chihara_1978}, there exists a unique linear functional ${\mathscr L}:\mathbb{R}[z]\to\mathbb{R}$ such that:
\begin{equation*}
\begin{aligned}
&{\mathscr L}[z^j\phi_i(z)]=h_i\delta_{i,j},\quad i=0,1,\ldots,N-1,\quad j=0,1,\ldots,i,\\
&{\mathscr L}[z^j\phi_N(z)]=0,\quad j=0,1,\ldots,
\end{aligned}
\end{equation*}
where $\delta_{i,j}$ is the Kronecker delta and $h_i$ are given using $b_1,b_2,\ldots,b_{N-1}$ and $h_0>0$ according to:
\[h_i\coloneqq (b_1b_2\cdots b_i)^2h_0,\quad i=1,2,\ldots,N-1.\]
Since $b_1^2,b_2^2,\ldots,b_{N-1}^2$ and $h_0$ are positive, it is obvious that ${\mathscr L}$ is positive definite. 

We now introduce the moment sequence $\{\mu_i\}$ with respect to ${\mathscr L}$ as:
\[\mu_i\coloneqq{\mathscr L}[z^i],\quad i=0,1,\ldots.\]
Then, we obtain determinantal expressions of the orthogonal polynomials $\phi_i(z)$ using the moments $\mu_i$. 
\begin{proposition}[cf.~\cite{Chihara_1978}]
The orthogonal polynomials $\phi_i(z)$ can be expressed in determinantal form as:
\[\phi_i(z)=\frac{1}{\tau_i}
\left\vert\begin{array}{ccccc}
\mu_0 & \mu_1 & \cdots & \mu_{i-1} & \mu_i \\
\mu_1 & \mu_2 & \cdots & \mu_i & \mu_{i+1} \\
\vdots & \vdots & \ddots & \vdots & \vdots \\
\mu_{i-1} & \mu_i & \cdots & \mu_{2i-2} & \mu_{2i-1} \\
1 & z & \cdots & z^{i-1} & z^i
\end{array}\right\vert,\quad i=1,2,\dots,N,\]
where $\tau_i$ are the Hankel determinants of degree $i$ given by:
\begin{equation*}
\begin{aligned}
&\tau_{-1}\coloneqq0,\quad\tau_0\coloneqq1,\\
&\tau_i\coloneqq
\left\vert\begin{array}{cccc}
\mu_0 & \mu_1 & \cdots & \mu_{i-1} \\
\mu_1 & \mu_2 & \cdots & \mu_i  \\
\vdots & \vdots & \ddots & \vdots \\
\mu_{i-1} & \mu_i & \cdots & \mu_{2i-2} 
\end{array}\right\vert,\quad i=1,2,\dots,N+1.
\end{aligned}
\end{equation*}
\end{proposition}
We denote distinct roots of the orthogonal polynomial $\phi_N(z)$ by $\lambda_1,\lambda_2,\ldots,\lambda_N$, namely:
\[\phi_N(z)=(z-\lambda_1)(z-\lambda_2)\cdots(z-\lambda_N).\]
Using the Gauss quadrature formula~\cite{Chihara_1978}, we obtain the following proposition concerning the positive definite $\mathscr{L}$ in terms of the roots $\lambda_1,\lambda_2,\dots,\lambda_N$. 
\begin{proposition}[cf.~\cite{Chihara_1978}]\label{prop:Gauss_quadrature_formula}
For all polynomials $P(z)$ of degree at most $2N-1$:
\[\mathscr{L}[P(z)]=\sum_{i=1}^Nc_iP(\lambda_i),\]
where $c_i$ are positive constants given by:
\[c_i\coloneqq\mathscr{L}\left[ 
\frac{\phi_N(z)}{\displaystyle \frac{d \phi_N(z)}{dz}
\Big\vert_{z=\lambda_i}(z-\lambda_i)}
\right],\quad i=1,2,\ldots,N.\]
\end{proposition}
\noindent
Proposition \ref{prop:Gauss_quadrature_formula} implies that we can uniquely determine the values of $\mu_1,\mu_2,\dots,\mu_{2N-1}$ from those of $\lambda_1,\lambda_2,\ldots,\lambda_N$ and $c_1,c_2,\ldots,c_N$. 

Now, we introduce two discrete-time parameters, $s$ and $n$, into the linear functional:
\begin{equation*}
\begin{aligned}
&{\mathscr L}^{(s+1,n)}[\cdot]\coloneqq{\mathscr L}^{(s,n)}[z\cdot],\\
&{\mathscr L}^{(s,n+1)}[\cdot]\coloneqq{\mathscr L}^{(s,n)}[(z+\delta^{(n)})\cdot],
\end{aligned}
\end{equation*}
where ${\mathcal L}^{(0,0)}\coloneqq{\mathcal L}$, and $\delta^{(n)}$ are arbitrary parameters. 
The discrete-time evolution from $n$ to $n+1$ differs from the discrete-time evolution from $s$ to $s+1$ in that it depends on the value of $\delta^{(n)}$. 
We consider the moment sequence $\{\mu_i^{(s,n)}\}$ with respect to ${\mathscr L}^{(s,n)}$ as:
\[\mu_i^{(s,n)}\coloneqq{\mathscr L}^{(s,n)}[x^i],\quad i=0,1,\ldots.\]
Then, the moments $\mu_i^{(s,n)}$ satisfy the two kinds of discrete-time evolutions:
\begin{equation}\label{eqn:moment_discrete_time_evoltion}
\begin{aligned}
&\mu_i^{(s+1,n)}=\mu_{i+1}^{(s,n)},\quad i=0,1,\ldots,\\
&\mu_i^{(s,n+1)}=\mu_{i+1}^{(s,n)}+\delta^{(n)}\mu_i^{(s,n)},\quad i=0,1,\ldots.
\end{aligned}
\end{equation}
We emphasize here that the moment sequence $\{\mu_i^{(s,n)}\}$ is equivalent to that appearing in Maeda et al.~\cite{Maeda_2013}. That is, under the assumption that the Hankel determinants $\tau_i^{(s,n)}\coloneqq\vert\mu_{j+k}^{(s,n)}\vert_{j,k=0}^{i-1}$ are all nonzero, we can express the orthogonal polynomials $\phi_i^{(s,n)}$:
\begin{equation}\label{eqn:phi_determinant_new}
\begin{aligned}
&\phi_{-1}^{(s,n)}(z)\coloneqq0,\quad \phi_0^{(s,n)}(z)\coloneqq1,\\
&\phi_i^{(s,n)}(z)=\frac{1}{\tau_i^{(s,n)}}
\left\vert\begin{array}{ccccc}
\mu_0^{(s,n)} & \mu_1^{(s,n)} & \cdots & \mu_{i-1}^{(s,n)} & \mu_i^{(s,n)} \\
\mu_1^{(s,n)} & \mu_2^{(s,n)} & \cdots & \mu_i^{(s,n)} & \mu_{i+1}^{(s,n)} \\
\vdots & \vdots & \ddots & \vdots & \vdots \\
\mu_{i-1}^{(s,n)} & \mu_i^{(s,n)} & \cdots & \mu_{2i-2}^{(s,n)} & \mu_{2i-1}^{(s,n)} \\
1 & z & \cdots & z^{i-1} & z^i
\end{array}\right\vert,\quad i=1,2,\dots,N.
\end{aligned}
\end{equation}
Moreover, we can derive the discrete-time evolutions from $s$ to $s+1$ and from $n$ to $n+1$ in the orthogonal polynomials $\phi_i^{(s,n)}(z)$ and their inverses. 
\begin{proposition}[cf.~\cite{Maeda_2013}]\label{prop:spectral_transformation}
Assume that $\tau_i^{(s,n)}\ne0$ for $i=1,2,\ldots,N.$ Then, the orthogonal polynomials $\phi_i^{(s,n)}(z)$ satisfy:
\begin{equation}\label{eqn:LR}
\left\{\begin{aligned}
&z\phi_{i-1}^{(s+1,n)}(z)=\phi_i^{(s,n)}(z)+Q_i^{(s,n)}\phi_{i-1}^{(s,n)}(z),\quad i=1,2,\dots,N,\\
&\phi_i^{(s,n)}(z)=\phi_i^{(s+1,n)}(z)+E_i^{(s,n)}\phi_{i-1}^{(s+1,n)}(z),\quad i=0,1,\dots,N,
\end{aligned}\right.
\end{equation}
and:
\begin{equation}\label{eqn:shiftedLR}
\left\{\begin{aligned}
&(z+\delta^{(n)})\phi_{i-1}^{(s,n+1)}(z)=\phi_i^{(s,n)}(z)+\bar{Q}_i^{(s,n)}\phi_{i-1}^{(s,n)}(z),\quad i=1,2,\dots,N,\\
&\phi_i^{(s,n)}(z)=\phi_i^{(s,n+1)}(z)+\bar{E}_i^{(s,n)}\phi_{i-1}^{(s,n+1)}(z),\quad i=0,1,\dots,N, 
\end{aligned}\right.
\end{equation}
where:
\begin{equation}\label{eqn:determinantal_solution_QE}
\left\{\begin{aligned}
&Q_i^{(s,n)}\coloneqq\frac{\tau_{i-1}^{(s,n)}\tau_i^{(s+1,n)}}{\tau_i^{(s,n)}\tau_{i-1}^{(s+1,n)}},\quad\bar{Q}_i^{(s,n)}\coloneqq\frac{\tau_{i-1}^{(s,n)}\tau_i^{(s,n+1)}}{\tau_i^{(s,n)}\tau_{i-1}^{(s,n+1)}},\quad i=1,2,\dots,N,\\
&E_i^{(s,n)}\coloneqq\frac{\tau_{i+1}^{(s,n)}\tau_{i-1}^{(s+1,n)}}{\tau_i^{(s,n)}\tau_i^{(s+1,n)}},\quad\bar{E}_i^{(s,n)}\coloneqq\frac{\tau_{i+1}^{(s,n)}\tau_{i-1}^{(s,n+1)}}{\tau_i^{(s,n)}\tau_i^{(s,n+1)}},\quad i=0,1,\dots,N.
\end{aligned}\right.
\end{equation}
\end{proposition}
\noindent
Except for the inclusion of discrete-time $n$, the first and second equations of \eqref{eqn:LR} are the Christoffel-Darboux and Geronimus transformations, respectively. 
Of course, in the absence of the discrete-time $s$, the first and second equations of \eqref{eqn:shiftedLR} are also the Christoffel-Darboux and Geronimus transformations, respectively. 
Equation \eqref{eqn:phi_determinant_new} with the boundary condition $\mathscr{L}^{(s,n)}[z^N\phi_N^{(s,n)}(z)]$ $=0$ immediately leads to $\tau_{N+1}^{(s,n)}/\tau_N^{(s,n)}=0$. Thus, it follows from \eqref{eqn:determinantal_solution_QE} that $E_N^{(s,n)}=0$ and $\bar{E}_N^{(s,n)}=0$. Combining these with the second equations of \eqref{eqn:LR} and \eqref{eqn:shiftedLR}, we observe that $\phi_N^{(s,n)}(z)=\phi_N^{(0,0)}(z)$. 
Thus, equations \eqref{eqn:LR} and \eqref{eqn:shiftedLR} are spectral transformations that do not change the values of $\lambda_1,\lambda_2,\ldots,\lambda_N$.

Equations \eqref{eqn:LR} and \eqref{eqn:shiftedLR} yield relationships of $\phi_i^{(s,n)}(z)$ to $\phi_{i-1}^{(s,n)}(z)$ and $\phi_{i-2}^{(s,n)}(z)$. 
\begin{lemma}[cf.~\cite{Maeda_2013}]\label{lem:compatibility_conditions}
The orthogonal polynomials $\phi_i^{(s,n)}$ satisfy the following three-term relationship:
\begin{equation}\label{eqn:orthogonal_polynomial}
z\phi_{i-1}^{(s,n)}(z)=\phi_i^{(s,n)}(z)+a_i^{(s,n)}\phi_{i-1}^{(s,n)}(z)+(b_{i-1}^{(s,n)})^2\phi_{i-2}^{(s,n)}(z),\quad i=1,2,\ldots,N,
\end{equation}
where $a_i^{(s,n)}$ and $b_i^{(s,n)}$ are given by:
\begin{eqnarray}
\label{eqn:compatibility_conditions_diag}
&a_i^{(s,n)}\coloneqq
\left\{\begin{aligned}
&Q_i^{(s,n)}+E_{i-1}^{(s,n)},\\
&Q_i^{(s-1,n)}+E_i^{(s-1,n)},\\
&\bar{Q}_i^{(s,n)}+\bar{E}_{i-1}^{(s,n)}-\delta^{(n)},\\
&\bar{Q}_i^{(s,n-1)}+\bar{E}_i^{(s,n-1)}-\delta^{(n-1)},
\end{aligned}\right.\quad i=1,2,\dots,N,
\\
\label{eqn:compatibility_conditions_subdiag}
&b_0^{(s,n)}\coloneqq0,\quad(b_i^{(s,n)})^2\coloneqq
\left\{\begin{aligned}
&Q_i^{(s,n)}E_i^{(s,n)},\\
&Q_{i+1}^{(s-1,n)}E_i^{(s-1,n)},\\
&\bar{Q}_i^{(s,n)}\bar{E}_i^{(s,n)},\\
&\bar{Q}_{i+1}^{(s,n-1)}\bar{E}_i^{(s,n-1)},
\end{aligned}\right.\quad i=1,2,\dots,N-1.
\end{eqnarray}
\end{lemma}
\noindent
Focusing on the first and second equations of \eqref{eqn:compatibility_conditions_diag} and \eqref{eqn:compatibility_conditions_subdiag}, we can obtain the discrete-time evolution from $s-1$ to $s$ of $Q_i^{(s,n)}$ and $E_i^{(s,n)}$. This is equivalent to that in the famous dToda equation. A nonautonomous version involving the parameters $\delta^{(n)}$ of the dToda equation can be derived from the third and fourth equations of \eqref{eqn:compatibility_conditions_diag} and \eqref{eqn:compatibility_conditions_subdiag}. 

Considering that $\bar{Q}_i^{(s,n)}$ and $\bar{E}_i^{(s,n)}$ are auxiliary variables in Lemma \ref{lem:compatibility_conditions}, we thus have recursion formulas that generate two-direction discrete-time evolutions.  
\begin{theorem}\label{thm:compatibility_condition_QE}
A combination of the discrete-time evolutions from $s$ to $s-1$ and from $n$ to $n+1$ of  the variables $Q_i^{(s,n)}$ and $E_i^{(s,n)}$ is given by:
\begin{equation}\label{eqn:compatibility_condition_1}
\left\{\begin{aligned}
&\bar{Q}_i^{(s,n)}+\bar{E}_{i-1}^{(s,n)}-\delta^{(n)}=Q_i^{(s,n)}+E_{i-1}^{(s,n)},\quad i=1,2,\dots,N,\\
&\bar{Q}_i^{(s,n)}\bar{E}_i^{(s,n)}=Q_i^{(s,n)}E_i^{(s,n)},\quad i=1,2,\dots,N-1,
\end{aligned}\right.
\end{equation}
and:
\begin{equation}\label{eqn:compatibility_condition_2}
\left\{\begin{aligned}
&Q_i^{(s-1,n+1)}+E_i^{(s-1,n+1)}=\bar{Q}_i^{(s,n)}+\bar{E}_i^{(s,n)}-\delta^{(n)},\quad i=1,2,\dots,N,\\
&Q_{i+1}^{(s-1,n+1)}E_i^{(s-1,n+1)}=\bar{Q}_{i+1}^{(s,n)}\bar{E}_i^{(s,n)},\quad i=1,2,\dots,N-1,
\end{aligned}\right.
\end{equation}
where $E_0^{(s,n)}\coloneqq0,E_N^{(s,n)}\coloneqq0$ and $\bar{E}_0^{(s,n)}\coloneqq0,\bar{E}_N^{(s,n)}\coloneqq0$. 
\end{theorem}
\noindent
Theorem \ref{thm:compatibility_condition_QE}, specifically its of $Q_i^{(s,n)},E_i^{(s,n)},\bar{Q}_i^{(s,n)}$ and $\bar{E}_i^{(s,n)}$, plays a key role in time-discretizing the peakon equation \eqref{eqn:peakon} in the following section. 
%
%
\section{$\mbox{\boldmath $LR$}$-like transformations and the discrete peakon equation}
\label{sec:4}
%
In this section, we clarify $LR$ transformations related to \eqref{eqn:compatibility_condition_1} and \eqref{eqn:compatibility_condition_2}, and then, using $LR$ transformations, we find a time-discretization version of the peakon equation \eqref{eqn:peakon}. 

For simplicity, we set $Q_i^{(n)}\coloneqq Q_i^{(-n,n)},E_i^{(n)}\coloneqq E_i^{(-n,n)},\bar{Q}_i^{(n)}\coloneqq\bar{Q}_i^{(-n,n)}$ and $\bar{E}_i^{(n)}\coloneqq\bar{E}_i^{(-n,n)}$. 
We assume that $\tau_i^{(s,n)}>0$ for $i=1,2,\ldots,N$ in \eqref{eqn:determinantal_solution_QE}. From this assumption, it is obvious that $Q_i^{(n)}>0,E_i^{(n)}>0,\bar{Q}_i^{(n)}>0$ and $\bar{E}_i^{(n)}>0$. We prepare $N$-by-$N$ upper and lower bidiagonal matrices:
\begin{equation*}
\begin{aligned}
&L^{(n)}\coloneqq 
\left( \begin{array}{cccc}
\sqrt{Q_1^{(n)}} & \\  -\sqrt{E_1^{(n)}} & \sqrt{Q_2^{(n)}} & \\
& \ddots & \ddots & \\ & & -\sqrt{E_{N-1}^{(n)}} & \sqrt{Q_N^{(n)}}
\end{array} \right),\quad R^{(n)}\coloneqq(L^{(n)})^\top,\\
&\bar{L}^{(n)}\coloneqq 
\left( \begin{array}{cccc}
\sqrt{\bar{Q}_1^{(n)}} & \\ -\sqrt{\bar{E}_1^{(n)}} & \sqrt{\bar{Q}_2^{(n)}} & \\
 & \ddots & \ddots & \\ & &  -\sqrt{\bar{E}_{N-1}^{(n)}} & \sqrt{\bar{Q}_N^{(n)}}
\end{array} \right),\quad \bar{R}^{(n)}\coloneqq (\bar{L}^{(n)})^\top.
\end{aligned}
\end{equation*}
This allows us to determine matrix representations of \eqref{eqn:compatibility_condition_1} and \eqref{eqn:compatibility_condition_2} as per the following Lemma. 
\begin{lemma}\label{lem:compatibility_condition_matrix}
Equations \eqref{eqn:compatibility_condition_1} and \eqref{eqn:compatibility_condition_2} can be expressed in matrix-form as:
\begin{equation}\label{eqn:compatibility_condition_matrix}
\left\{\begin{aligned}
&\bar{L}^{(n)}\bar{R}^{(n)}-\delta^{(n)}I=L^{(n)}R^{(n)},\\
&R^{(n+1)}L^{(n+1)}=\bar{R}^{(n)}\bar{L}^{(n)}-\delta^{(n)}I,
\end{aligned}\right.
\end{equation}
where $I$ is the $N$-by-$N$ identity matrix. 
\end{lemma}
\begin{proof}	
The $(i,i)$ and $(i,i+1)$ entries of $\bar{L}^{(n)}\bar{R}^{(n)}-\delta^{(n)}I$ are $\bar{Q}_i^{(s,n)}+\bar{E}_{i-1}^{(s,n)}-\delta^{(n)}$ and $-\sqrt{\bar{Q}_i^{(s,n)}\bar{E}_i^{(s,n)}}$, respectively. Using \eqref{eqn:compatibility_condition_1}, we can rewrite the $(i,i)$ and $(i,i+1)$ entries of $\bar{L}^{(n)}\bar{R}^{(n)}-\delta^{(n)}I$ as $Q_i^{(s,n)}+E_{i-1}^{(s,n)}$ and $-\sqrt{Q_i^{(s,n)}E_i^{(s,n)}}$, respectively. 
These are simply the $(i,i)$ and $(i,i+1)$ entries of $L^{(n)}R^{(n)}$. Since $\bar{L}^{(n)}\bar{R}^{(n)}-\delta^{(n)}I$ and $L^{(n)}R^{(n)}$ are both symmetric tridiagonal, the expression in the $(i+1,i)$ entry is equal to that in the $(i,i+1)$ entry.  
We thus obtain the first equation of \eqref{eqn:compatibility_condition_matrix}. Similarly, by focusing on the tridiagonal parts of $R^{(n+1)}L^{(n+1)}$ and $R^{(n)}L^{(n)}-\delta^{(n)}I$, and using \eqref{eqn:compatibility_condition_2}, we obtain the second equation of \eqref{eqn:compatibility_condition_matrix}. 
\end{proof}

We remark that \eqref{eqn:compatibility_condition_matrix} can be regarded as the implicit-shift $LR$ transformation from $L^{(n)}R^{(n)}$ to $R^{(n+1)}L^{(n+1)}$. 
Since ${\rm det}\,L^{(n)}\ne0$ and ${\rm det}\,R^{(n)}\ne0$, the inverses $(L^{(n)})^{-1}$ and $(R^{(n)})^{-1}$ exist. We introduce $N$-by-$N$ lower triangular ${\cal L}^{(n)}\coloneqq(L^{(n)})^{-1}\bar{L}^{(n)}$ and upper triangular ${\cal R}^{(n)}\coloneqq\bar{R}^{(n)}(R^{(n)})^{-1}$. 
Observing Lemma \ref{lem:compatibility_condition_matrix} in terms of tridiagonal $A^{(n)}\coloneqq R^{(n)}L^{(n)}$, we can find the similarity transformation from $A^{(n)}$ to $A^{(n+1)}$. 
\begin{proposition}\label{prop:Lax_representation_dpeakon_Toda}
The discrete-time evolution from $n$ to $n+1$ of $A^{(n)}$ can be expressed as:
\begin{equation}\label{eqn:Lax_representation_dpeakon_Toda}
A^{(n+1)}=({\cal L}^{(n)})^{-1}A^{(n)}{\cal L}^{(n)}={\cal R}^{(n)}A^{(n)}({\cal R}^{(n)})^{-1}.
\end{equation}
\end{proposition}
\begin{proof}
Similar to $L^{(n)}$ and $R^{(n)}$, the inverse $(\bar{R}^{(n)})^{-1}$ exists.
The second equation of \eqref{eqn:compatibility_condition_matrix} immediately leads to:
\[R^{(n+1)}L^{(n+1)}=\bar{R}^{(n)}(\bar{L}^{(n)}\bar{R}^{(n)}-\delta^{(n)}I)(\bar{R}^{(n)})^{-1}.\]
Combining this with the first equation of \eqref{eqn:compatibility_condition_matrix}, we derive:
\[R^{(n+1)}L^{(n+1)}=\bar{R}^{(n)}(L^{(n)}R^{(n)})(\bar{R}^{(n)})^{-1}.\]
Noting that $R^{(n+1)}L^{(n+1)}=A^{(n+1)}$ and $\bar{R}^{(n)}(L^{(n)}R^{(n)})(\bar{R}^{(n)})^{-1}=(\bar{R}^{(n)}(R^{(n)})^{-1})(R^{(n)}L^{(n)})$ $(R^{(n)}(\bar{R}^{(n)})^{-1})={\cal R}^{(n)}A^{(n)}({\cal R}^{(n)})^{-1}$, we can derive $A^{(n+1)}={\cal R}^{(n)}A^{(n)}({\cal R}^{(n)})^{-1}$. 
Moreover, by replacing $\bar{R}^{(n)}$ and $(\bar{R}^{(n)})^{-1}$ with $(\bar{L}^{(n)})^{-1}$ and $\bar{L}^{(n)}$, respectively, in the above discussion, we have $A^{(n+1)}=({\cal L}^{(n)})^{-1}A^{(n)}{\cal L}^{(n)}$. 
\end{proof}

Using ${\cal L}^{(n)}$ and ${\cal R}^{(n)}$, we can derive the $LR$-like transformation with an implicit shift, which is simplified to simply be the $LR$ transformation if $\delta^{(n+1)}=\delta^{(n)}$. 
\begin{proposition}\label{prop:Lax_representation_dpeakon}
For ${\cal L}^{(n)}$ and ${\cal R}^{(n)}$, it holds that:
\begin{equation}\label{eqn:Lax_representation_dpeakon}
\frac{1}{\delta^{(n+1)}} \left( {\cal L}^{(n+1)} {\cal R}^{(n+1)}-I \right)=\frac{1}{\delta^{(n)}} \left( {\cal R}^{(n)}{\cal L}^{(n)}-I \right),
\end{equation}
\end{proposition}
\begin{proof}
We prepare an $N$-by-$N$ matrix:
\[
{\cal A}^{(n)}\coloneqq \frac{1}{\delta^{(n)}} \left( {\cal L}^{(n)} {\cal R}^{(n)}-I \right).
\]
Recalling that ${\cal L}^{(n)}=(L^{(n)})^{-1}\bar{L}^{(n)}$ and ${\cal R}^{(n)}=\bar{R}^{(n)}(R^{(n)})^{-1}$, and using the first equation of \eqref{eqn:compatibility_condition_matrix} with $R^{(n)}L^{(n)}=A^{(n)}$, we see that:
\[
{\cal A}^{(n)}=(A^{(n)})^{-1}.
\]
Combining this with Proposition \ref{prop:Lax_representation_dpeakon_Toda}, we obtain:
\[
{\cal A}^{(n+1)}=({\cal L}^{(n)})^{-1}{\cal A}^{(n)}{\cal L}^{(n)}={\cal R}^{(n)}{\cal A}^{(n)}({\cal R}^{(n)})^{-1}.
\]
This similarity transformation immediately leads to \eqref{eqn:Lax_representation_dpeakon}. 
\end{proof}

Now, we introduce new variables $p_i^{(n)}$ and $x_i^{(n)}$ which satisfy:
\begin{equation}\label{eqn:QE_and_dpeakon_transformation}
\left\{\begin{aligned}
&1+\delta^{(n)}p_i^{(n)}\coloneqq\frac{\bar{Q}_i^{(n)}}{Q_i^{(n)}},\quad i=1,2,\dots,N,\\
&e^{ x_{i+1}^{(n)}-x_i^{(n)} }\coloneqq\frac{ p_i^{(n)}E_i^{(n)} }{ p_{i+1}^{(n)}Q_{i+1}^{(n)}},\quad i=1,2,\dots,N-1.
\end{aligned}\right.
\end{equation}
These are variables that will later be shown to be time-discretizations 
of the peakon variables $p_i(t)$ and $x_i(t)$. 
From the second equation of \eqref{eqn:compatibility_condition_1} and the first equation of \eqref{eqn:QE_and_dpeakon_transformation}, it follows that:
\begin{equation}\label{eqn:QE_and_dpeakon_transformation_sub}
\frac{\bar{E}_i^{(n)}}{E_i^{(n)}}=\frac{1}{1+\delta^{(n)}p_i^{(n)}},\quad i=1,2,\ldots,N-1.
\end{equation}
Using \eqref{eqn:QE_and_dpeakon_transformation} and \eqref{eqn:QE_and_dpeakon_transformation_sub}, we can rewrite the entries of ${\cal R}^{(n)}=\bar{R}^{(n)}(R^{(n)})^{-1}$ in terms of $p_i^{(n)}$, $x_i^{(n)}$ and $\delta^{(n)}$ as follows:
\begin{equation}\label{eqn:dpeakon_matrix}
{\cal R}^{(n)}=\left( \begin{array}{cccc}
\sqrt{1+\delta^{(n)}p_1^{(n)}} 
& \displaystyle\frac{\delta^{(n)}p_1^{(n)}}{\sqrt{1+\delta^{(n)}p_1^{(n)}} }\sqrt{ \frac{ p_2^{(n)}e^{x_2^{(n)}} }{ p_1^{(n)}e^{x_1^{(n)}} } } & \cdots 
& \displaystyle\frac{\delta^{(n)}p_1^{(n)}}{\sqrt{1+\delta^{(n)}p_1^{(n)}} }\sqrt{ \frac{ p_N^{(n)}e^{x_N^{(n)}} }{ p_1^{(n)}e^{x_1^{(n)}} } } \\
& \sqrt{1+\delta^{(n)}p_2^{(n)}} & \cdots & 
\displaystyle\frac{\delta^{(n)}p_2^{(n)}}{\sqrt{1+\delta^{(n)}p_2^{(n)}} }
\sqrt{ \displaystyle\frac{ p_N^{(n)}e^{x_N^{(n)}} }{ p_2^{(n)}e^{x_2^{(n)}} } } \\
& & \ddots & \vdots \\& & & \sqrt{1+\delta^{(n)}p_{N}^{(n)}}
\end{array} \right).
\end{equation}
Of course, ${\cal L}^{(n)}$ can be rewritten as the transpose of this matrix. 
Proposition \ref{prop:Lax_representation_dpeakon} thus enables the design of a recursion formula that generates the discrete-time evolution from $n$ to $n+1$ of $p_i^{(n)}$ and $x_i^{(n)}$. 
\begin{theorem}\label{thm:dpeakon}
The discrete-time evolution from $n$ to $n+1$ of $p_i^{(n)}$ and $x_i^{(n)}$ can be expressed involving $\delta^{(n)}$ as follows:
\begin{equation}\label{eqn:dpeakon}
\left\{\begin{aligned}
&p_i^{(n+1)}+\sum_{j=1}^{i-1}\frac{\delta^{(n+1)}p_j^{(n+1)}p_i^{(n+1)}e^{x_i^{(n+1)}-x_j^{(n+1)}} }{ 1+\delta^{(n+1)}p_j^{(n+1)}}
=p_i^{(n)}+\sum_{j=i+1}^N\frac{ \delta^{(n)}p_i^{(n)}p_j^{(n)}e^{x_j^{(n)}-x_i^{(n)}} }{1+\delta^{(n)}p_i^{(n)}},\\
&\quad i=1,2,\ldots,N, \\
&e^{x_i^{(n+1)}-x_i^{(n)}}=\frac{p_i^{(n)}}{p_i^{(n+1)}}
\frac{\displaystyle \left( 1+\delta^{(n)} \sum_{j=i}^{N}p_j^{(n)}e^{x_j^{(n)}-x_i^{(n)}} \right)^2 }{ 1+\delta^{(n)}p_i^{(n)} },\quad i=1,2,\ldots,N,
\end{aligned}\right.
\end{equation}
\end{theorem}
\begin{proof}
Observing that the $(i,i)$ entries of $(1/\delta^{(n+1)})( {\cal L}^{(n+1)}{\cal R}^{(n+1)}-I )=(1/\delta^{(n)})( {\cal R}^{(n)}{\cal L}^{(n)}-I )$, we can easily derive the first equation of \eqref{eqn:dpeakon}. Focusing on the $(i,j)$ entries of $(1/\delta^{(n+1)})( {\cal L}^{(n+1)}{\cal R}^{(n+1)}-I )=(1/\delta^{(n)})( {\cal R}^{(n)}{\cal L}^{(n)}-I )$, we obtain:
\begin{eqnarray}
&\nonumber
\left( p_j^{(n+1)}+\sum_{k=1}^{j-1} 
\frac{ \delta^{(n+1)}(p_k^{(n+1)})^2 }{ 1+\delta^{(n+1)}p_k^{(n+1)} }\frac{ p_j^{(n+1)}e^{x_j^{(n+1)}} }{ p_k^{(n+1)}e^{x_k^{(n+1)}}}  \right)\sqrt{ \frac{ p_i^{(n+1)}e^{x_i^{(n+1)}} }{ p_j^{(n+1)}e^{x_j^{(n+1)}} } } \\
&\nonumber
\quad=\left( p_j^{(n)}+\sum_{k=i+1}^N\frac{ \delta^{(n)}p_j^{(n)}p_i^{(n)} }{1+\delta^{(n)}p_i^{(n)}} 
\frac{ p_k^{(n)}e^{x_k^{(n)}} }{ p_i^{(n)}e^{x_i^{(n)}} } \right)
\sqrt{ \frac{ 1+\delta^{(n)}p_i^{(n)} }{ 1+\delta^{(n)}p_j^{(n)} }
\frac{ p_i^{(n)}e^{x_i^{(n)}} }{ p_j^{(n)}e^{x_j^{(n)}} } },\\
&\qquad i=1,2,\ldots,N,\quad j=1,2,\ldots,i-1.
\label{eqn:dpeakon_proof1}
\end{eqnarray}
From the first equation of \eqref{eqn:dpeakon} with \eqref{eqn:dpeakon_proof1}, it follows that:
\begin{eqnarray*}
&\frac{ \displaystyle p_j^{(n)} \left( 1+\delta^{(n)} \sum_{k=j}^{N}p_k^{(n)}e^{x_k^{(n)}-x_j^{(n)}} \right) }{1+\delta^{(n)}p_j^{(n)}} \sqrt{ \frac{ p_i^{(n+1)}e^{x_i^{(n+1)}} }{ p_j^{(n+1)}e^{x_j^{(n+1)}} } } \\
&\quad=\frac{ \displaystyle p_j^{(n)} \left( 1+\delta^{(n)} \sum_{k=i}^{N}p_k^{(n)}e^{x_k^{(n)}-x_i^{(n)}} \right) }{1+\delta^{(n)}p_i^{(n)}}\sqrt{ \frac{ 1+\delta^{(n)}p_i^{(n)} }{ 1+\delta^{(n)}p_j^{(n)} } \frac{ p_i^{(n)}e^{x_i^{(n)}} }{ p_j^{(n)}e^{x_j^{(n)}} } }.
\end{eqnarray*}
We can rewrite this as:
\begin{equation}\label{eqn:dpeakon_proof2}
\begin{aligned}
&\frac{ p_i^{(n+1)}e^{x_i^{(n+1)}} }{ p_i^{(n)}e^{x_i^{(n)}} }
\frac{ 1+\delta^{(n)}p_i^{(n)} }{ \displaystyle \left(1+\delta^{(n)} \sum_{k=i}^Np_k^{(n)}e^{x_k^{(n)}-x_i^{(n)}} \right)^2}
=\frac{ p_j^{(n+1)}e^{x_j^{(n+1)}} }{ p_j^{(n)}e^{x_j^{(n)}} }
\frac{ 1+\delta^{(n)}p_j^{(n)} }{ \displaystyle \left(1+\delta^{(n)} \sum_{k=j}^Np_k^{(n)}e^{x_k^{(n)}-x_j^{(n)}} \right)^2}.
\end{aligned}
\end{equation}
We emphasize that the variable subscripts on the lefthand side of \eqref{eqn:dpeakon_proof2} are all $i$ and those on the righthand side of \eqref{eqn:dpeakon_proof2} are all $j$. 
This implies that the both sides of \eqref{eqn:dpeakon_proof2} do not depend on the variable subscripts, namely, for some constant $C\neq0$, it holds that:
\begin{equation}\label{eqn:dpeakon_proof3}
\frac{ p_i^{(n+1)}e^{x_i^{(n+1)}} }{ p_i^{(n)}e^{x_i^{(n)}} }
\frac{ 1+\delta^{(n)}p_i^{(n)} }{ \displaystyle \left(1+\delta^{(n)} \sum_{k=i}^Np_k^{(n)}e^{x_k^{(n)}-x_i^{(n)}} \right)^2}=C,\quad i=1,2,\ldots,N.
\end{equation}
Thus, by letting $C=1$ in \eqref{eqn:dpeakon_proof3}, we obtain the second equation of \eqref{eqn:dpeakon}. 
\end{proof}

The following aims to show that \eqref{eqn:dpeakon} can be regarded as a time-discretization of the peakon equation \eqref{eqn:peakon}. 
The first equation of \eqref{eqn:dpeakon} leads to:
\begin{equation}\label{eqn:dpeakon_limit_poof_1}
\frac{p_i^{(n+1)}-p_i^{(n)}}{\delta^{(n)}}=-\frac{\delta^{(n+1)}}{\delta^{(n)}}
\sum_{j=1}^{i-1}\frac{ p_j^{(n+1)} p_i^{(n+1)}e^{x_i^{(n+1)}-x_j^{(n+1)}} }{ 1+\delta^{(n+1)}p_j^{(n+1)}}+\sum_{j=i+1}^N\frac{ p_j^{(n)}p_i^{(n)}e^{x_j^{(n)}-x_i^{(n)}} }{1+\delta^{(n)}p_i^{(n)}}.
\end{equation}
We here define continuous-time $t$ as $t=\delta^{(0)}+\delta^{(1)}+\cdots+\delta^{(n-1)}$. Then, in \eqref{eqn:dpeakon_limit_poof_1}, the discrete-time $n$ is related to the continuous-time $t$ such that $p_i^{(n)}=p_i(t),x_i^{(n)}=x_i(t)$ and $p_i^{(n+1)}=p_i(t+\delta^{(n)})$, and $x_i^{(n+1)}=x_i(t+\delta^{(n)})$. 
Thus, by taking the limit as $\delta^{(n)}\to0$ in \eqref{eqn:dpeakon_limit_poof_1} where $\delta^{(n+1)}/\delta^{(n)}\to1$, we derive the first equation of \eqref{eqn:peakon}. 

Taking the logarithm of the both sides of the second equation of \eqref{eqn:dpeakon} and then dividing them by $\delta^{(n)}$, we obtain:
\[
\begin{aligned}
\frac{x_i^{(n+1)}-x_i^{(n)}}{\delta^{(n)}}
&=-\frac{\log p_i^{(n+1)}-\log p_i^{(n)} }{ \delta^{(n)} }
-\frac{ \log\left(1+\delta^{(n)}p_i^{(n)}\right) }{ \delta^{(n)} }
\\
&\quad+2\frac{ \displaystyle \log\left(1+\delta^{(n)}\sum_{j=i+1}^{N} p_j^{(n)}e^{x_j^{(n)}-x_i^{(n)}} \right) }{ \delta^{(n)} }.
\end{aligned}
\]
Replacing $p_i^{(n)},x_i^{(n)}$ and $p_i^{(n+1)},x_i^{(n+1)}$ with $p_i(t), x_i(t)$ and $p_i(t+\delta^{(n)}),x_i(t+\delta^{(n)})$, respectively, and taking the limit as $\delta^{(n)}\to0$, we derive:
\[\frac{d x_i(t) }{ dt }=-\frac{ d \log p_i(t) }{ dt }-p_i(t)+2\sum_{j=i+1}^{N} p_j(t)e^{ x_j(t)-x_i(t) },\]
which is equivalent to the second equation of \eqref{eqn:peakon}. 
Therefore, we have the following important theorem concerning time-discretization of the peakon equation \eqref{eqn:peakon}.

\begin{theorem}\label{thm:peakon_dpeakon}  Equation \eqref{eqn:dpeakon} is a time-discretization involving the discretization parameter $\delta^{(n)}$ of the peakon equation \eqref{eqn:peakon}.
\end{theorem}
Combining the first and second equations of \eqref{eqn:dpeakon}, we obtain:
\begin{equation*}
\begin{aligned}
p_i^{(n+1)}&=p_i^{(n)}+\sum_{j=i+1}^{N}\frac{ \delta^{(n)} p_j^{(n)}p_i^{(n)}e^{x_j^{(n)}-x_i^{(n)}} }{1+\delta^{(n)}p_i^{(n)}}\\
&\quad -\frac{\displaystyle \left(1+\delta^{(n)}\sum_{j=i+1}^N p_j^{(n)}e^{x_j^{(n)}-x_i^{(n)}} \right)^2}{1+\delta^{(n)}p_i^{(n)}}\sum_{j=1}^{i-1} \frac{\delta^{(n+1)} p_j^{(n+1)}p_i^{(n)} e^{x_i^{(n)}-x_j^{(n+1)}} }{ 1+\delta^{(n+1)} p_j^{(n+1)}}.
\end{aligned}
\end{equation*}
Thus, we see that the values of $p_i^{(n+1)}$ are uniquely determined from those of $\{ p_j^{(n)} \}_{j=i}^{N},\{ x_j^{(n)} \}_{j=i}^{N}$ and $\{ p_j^{(n+1)} \}_{j=1}^{i-1},$ $\{ x_j^{(n+1)} \}_{j=1}^{i-1}$ and the discretization parameters $\delta^{(n)}$ and $\delta^{(n+1)}$. 
It is clear that, in the second equation of \eqref{eqn:dpeakon}, the values of $x_i^{(n+1)}$ can be determined the same way. 
Figure~\ref{fig:diagram} illustrates how to generate discrete-time evolutions of the variables $p_i^{(n)}$ and $x_i^{(n)}$. 
\begin{figure}\begin{center}
\includegraphics[width=\textwidth]{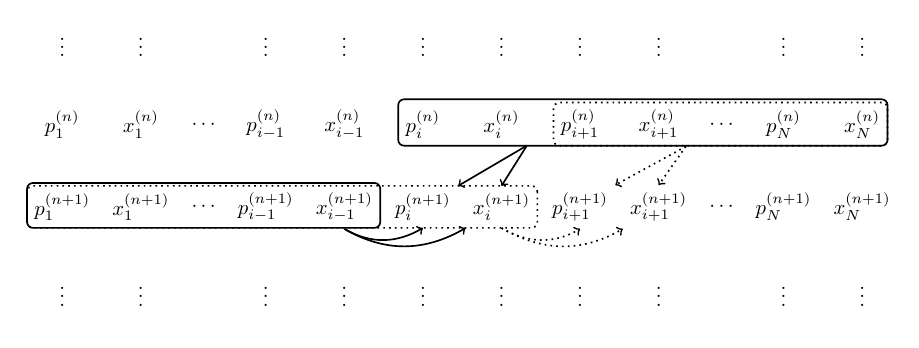}
\caption{Discrete-time evolutions of the variables $p_i^{(n)}$ and $x_i^{(n)}$.}\label{fig:diagram}
\end{center}
\end{figure}

The $LR$-like perspective plays a key role in concluding that \eqref{eqn:dpeakon} is a discrete peakon equation. To make this concept more convincing, we also consider the matrix representation of \eqref{eqn:peakon}. 
From the shifted $LR$-like transformation \eqref{eqn:Lax_representation_dpeakon} with ${\cal A}^{(n)}=(1/\delta^{(n)})( {\cal L}^{(n)} {\cal R}^{(n)}-I )$, it follows that:
\begin{equation}\label{eqn:dpeakon_matrix}
\begin{aligned}
&\frac{1}{\delta^{(n)}}\left({\cal A}^{(n+1)}-{\cal A}^{(n)}\right)\\
&\quad=
\frac{1}{ \delta^{(n)} } 
\left( ({\cal R}^{(n)})_{>} + \Omega({\cal R}^{(n)})-I \right) {\cal A}^{(n)}
- {\cal A}^{(n+1)} \frac{ 1 }{\delta^{(n)}} \left( ({\cal R}^{(n)})_{>}+\Omega({\cal R}^{(n)})-I \right),
\end{aligned}
\end{equation}
where $\Omega({\cal R}^{(n)})$ denotes the diagonal matrix whose diagonal entries are those of ${\cal R}^{(n)}$. 
Since ${\cal A}^{(n)}\to{\cal A}(t),{\cal A}^{(n+1)}\to{\cal A}(t),(1/\delta^{(n)})({\cal R}^{(n)})_{>}\to({\cal A}(t))_{>}$ and $(1/\delta^{(n)})(\Omega({\cal R}^{(n)})-I)\to(1/2)\Omega({\cal A}(t))$ as $\delta^{(n)}\to0$ in \eqref{eqn:dpeakon_matrix}, we obtain:
\[\frac{d{\cal A}(t)}{dt}
=\left( ({\cal A}(t))_{>}+\frac{1}{2}\Omega({\cal A}(t)) \right) {\cal A}(t)
-{\cal A}(t) \left( ({\cal A}(t))_{>}+\frac{1}{2}\Omega({\cal A}(t)) \right).\]
Moreover, by noting that $({\cal A}(t))_{>}+(1/2)\Omega({\cal A}(t))=(1/2){\cal A}(t)-\Pi({\cal A}(t))$, we can derive the matrix representation of the peakon equation \eqref{eqn:peakon},  namely, \eqref{eqn:Lax_representation_peakon}. 

Taking the limit as $\delta^{(n)}\to\infty$ in the shifted $LR$-like transformation \eqref{eqn:Lax_representation_dpeakon}, and setting $U^{(n)}\coloneqq \lim_{\delta^{(n)}\to\infty}R^{(n)}/\sqrt{\delta^{(n)}}$, we obtain:
\begin{equation}\label{eqn:LR_U}
(U^{(n+1)})^\top U^{(n+1)}=U^{(n)}(U^{(n)})^\top,
\end{equation}
where $U^{(n)}$ is an upper triangular matrix defined by:
\[
U^{(n)}\coloneqq\left( \begin{array}{cccc}
\sqrt{ p_1^{(n)} } & \sqrt{ p_2^{(n)} } e^{ \frac{x_2^{(n)}-x_1^{(n)}}{2} } & \cdots 
& \sqrt{ p_N^{(n)} } e^{ \frac{x_N^{(n)}-x_1^{(n)}}{2} }  \\
& \sqrt{ p_2^{(n)} } & \cdots & \sqrt{ p_N^{(n)} } e^{ \frac{x_N^{(n)}-x_2^{(n)}}{2} } \\
& & \ddots & \vdots \\& & & \sqrt{ p_N^{(n)} }
\end{array} \right).
\]
We pepare the variables $\bar{p}_i^{(n)}$ and $\bar{x}_i^{(n)}$ expressed using the discrete peakon variables $p_i^{(n)}$ and $x_i^{(n)}$ as follows:
\begin{equation*}
\left\{\begin{aligned}
&\bar{p}_1^{(n)}\coloneqq p_1^{(n)},\quad
\bar{p}_i^{(n)}\coloneqq \frac{p_i^{(n)}}{1-e^{\bar{x}_i^{(n)}-\bar{x}_{i-1}^{(n)}}},\quad i=2,3,\ldots,N,\\
&e^{\bar{x}_2^{(n)}-\bar{x}_1^{(n)}} \coloneqq \frac{ e^{x_2^{(n)}-x_1^{(n)}} }{ 1+e^{x_2^{(n)}-x_1^{(n)}} },
\\
&e^{\bar{x}_{i+1}^{(n)}-\bar{x}_i^{(n)}}\coloneqq \frac{ e^{x_{i+1}^{(n)}-x_i^{(n)}} }{ 1+e^{x_{i+1}^{(n)}-x_i^{(n)}}-e^{\bar{x}_i^{(n)}-\bar{x}_{i-1}^{(n)}} },\quad i=2,3,\ldots,N-1.
\end{aligned}\right.
\end{equation*}
Then, the $(i,j)$ entry of $U^{(n)}$ can be expressed using $\bar{p}_i^{(n)}$ and $\bar{x}_i^{(n)}$ as:
\[
(U^{(n)})_{i,j}=\sqrt{ \bar{p}_j^{(n)}(e^{\bar{x}_j^{(n)}-\bar{x}_i^{(n)}}-e^{\bar{x}_j^{(n)}-\bar{x}_{i-1}^{(n)}}) }.
\]
Thus, from \eqref{eqn:LR_U}, we can derive a recursion formula for generating the discrete-time evolution from $n$ to $n+1$ of $\bar{p}_i^{(n)}$ and $\bar{x}_i^{(n)}$. 
We can easily check that the resulting recursion formula is, in fact, equivalent to the peakon lattice equation presented by Ragnisco et al.~\cite{Ragnisco_1996}. 
%
%
%
\section{Determinantal solution and its asymptotic behavior}
\label{sec:5}
%
In this section, we present the determinantal solution to the discrete peakon equation \eqref{eqn:dpeakon}, and then clarify asymptotic behaviors as $n\to\infty$ of the discrete peakon variables $p_i^{(n)}$ and $x_i^{(n)}$. We also give numerical examples to demonstrate peakon waves motions. 

We begin by reconsidering the moments $\mu_i^{(s,n)}$ from Section \ref{sec:3}. 
Since $\mathcal{L}^{(0,0)}\coloneqq\mathcal{L}$, it is obvious that $\mu_i^{(0,0)}=\mu_i$. From Proposition \ref{prop:Gauss_quadrature_formula}, we thus see that the moments $\mu_i^{(0,0)}$ can be expressed using $\lambda_1,\lambda_2,\ldots,\lambda_N$ and $c_1,c_2,\ldots,c_N$ as:
\[
\mu_i^{(0,0)}=\sum_{j=1}^N c_j\lambda_j^i,\quad i=1,2,\ldots,2N,
\]
where $c_1>0,c_2>0,\dots,c _N>0$. 
Combining this with the moment recursion formula \eqref{eqn:moment_discrete_time_evoltion}, we derive:
\[
\mu_i^{(s,n)}=\sum_{j=1}^N c_j\lambda_j^{i+s}\prod_{\ell=0}^{n-1}(\lambda_j+\delta^{(\ell)}),\quad i=1,2,\ldots,2N.
\]
Along the same line as Kobayashi~\cite{Kobayashi_2021}, we expand the Hankel determinants $\tau_i^{(s,n)}$:
\begin{equation}\label{eqn:determinant_expand}
\begin{aligned}
&\tau_i^{(s,n)}=\sum_{1\leq j_1<j_2<\cdots<j_i\leq N}\left\{
\prod_{k=1}^i \left[ c_{j_k}\lambda_{j_k}^s \prod_{\ell=0}^{n-1}(\lambda_{j_k}+\delta^{(\ell)}) \right]\prod_{1\leq k<\ell\leq i}(\lambda_{j_k}-\lambda_{j_\ell})^2\right\},\\
&\quad i=1,2,\ldots,N.
\end{aligned}
\end{equation}
We can easily obtain a proposition regarding the positivity of the Hankel determinants $\tau_i^{(s,n)}$. 
\begin{proposition}\label{prop:existence_dpeakon}
Let us assume that $\lambda_1>0,\lambda_2>0,\ldots,\lambda_N>0$. If $\delta^{(n)}>-\min_i \lambda_i$ for any $n$, namely, $\lambda_i+\delta^{(n)}>0$ then it holds that $\tau_1^{(s,n)}>0,\tau_2^{(s,n)}>0,\ldots,\tau_N^{(s,n)}>0$. 
\end{proposition}
\noindent
From Proposition \ref{prop:existence_dpeakon} with \eqref{eqn:determinantal_solution_QE}, it follows that $Q_i^{(n)}>0,E_i^{(n)}>0$ and $\bar{Q}_i^{(n)}>0,\bar{E}_i^{(n)}>0$. 
Under this positivity, \eqref{eqn:QE_and_dpeakon_transformation} implies that the discrete peakon variables $p_i^{(n)}$ and $x_i^{(n)}$ are well-defined. 
Moreover, by combining \eqref{eqn:determinantal_solution_QE} with \eqref{eqn:QE_and_dpeakon_transformation}, we obtain determinantal expressions of the discrete peakon variables $p_i^{(n)}$ and $x_i^{(n)}$. 
\begin{theorem}\label{thm:solution_dpeakon}
Assume that $\lambda_1>0,\lambda_2>0,\ldots,\lambda_N>0$ and $\delta^{(0)}>-\min_i\lambda_i,\delta^{(1)}>-\min_i\lambda_i,\ldots.$ Then, the discrete peakon variables $p_i^{(n)}$ and $x_i^{(n)}$ can be expressed using the Hankel determinants as follows:
\begin{equation}\label{eqn:solution_dpeakon}
\left\{\begin{aligned}
&p_i^{(n)}=\frac{\tau_i^{(-n,n)}\tau_{i-1}^{(-n+1,n+1)}}{\tau_{i-1}^{(-n,n+1)}\tau_i^{(-n+1,n)}},\quad i=1,2,\dots,N,\\
&x_i^{(n)}=\log\frac{\tau_i^{(-n,n)}\tau_{i-1}^{(-n,n+1)}}{\tau_{i-1}^{(-n+1,n)}\tau_{i-1}^{(-n+1,n+1)}},\quad i=1,2,\dots,N.
\end{aligned}\right.
\end{equation}
\end{theorem}
\begin{proof}
Applying $z=0$ in Proposition \ref{prop:spectral_transformation}, we obtain:
\begin{equation*}\left\{\begin{aligned}
&0=\phi_i^{(s,n)}(0)+Q_i^{(s,n)}\phi_{i-1}^{(s,n)}(0),\quad i=1,2,\ldots,N,\\
&\delta^{(n)}\phi_{i-1}^{(s,n+1)}(0)=\phi_i^{(s,n)}(0)+\bar{Q}_i^{(s,n)}\phi_{i-1}^{(s,n)}(0),\quad i=1,2,\ldots,N,
\end{aligned}\right.\end{equation*}
This immediately leads to:
\begin{equation}\label{eqn:Q_Qbar}
1-\delta^{(n)}\frac{\phi_{i-1}^{(s,n+1)}(0)}{\phi_i^{(s,n)}(0)}=\frac{\bar{Q}_i^{(s,n)}}{Q_i^{(s,n)}},\quad i=1,2,\ldots,N.
\end{equation}
Recalling that $Q_i^{(n)}=Q_i^{(-n,n)}$ and $\bar{Q}_i^{(n)}=\bar{Q}_i^{(-n,n)}$, and combining the first equation of \eqref{eqn:QE_and_dpeakon_transformation} with \eqref{eqn:Q_Qbar}, we derive:
\begin{equation}\label{eqn:solution_dpeakon_proof1}
p_i^{(n)}=-\frac{\phi_{i-1}^{(-n,n+1)}(0)}{\phi_i^{(-n,n)}(0)},\quad i=1,2,\ldots,N.
\end{equation}
Moreover, by applying $z=0$ in \eqref{eqn:phi_determinant_new} and using the first equation of \eqref{eqn:moment_discrete_time_evoltion}, we obtain $\phi_i^{(s,n)}(0)=(-1)^i\tau_i^{(s+1,n)}/\tau_i^{(s,n)}$. Combining this with \eqref{eqn:solution_dpeakon_proof1}, we thus have the first equation of \eqref{eqn:solution_dpeakon}. 

Using \eqref{eqn:determinantal_solution_QE} and the first equation of \eqref{eqn:solution_dpeakon}, we can rewrite the second equation of \eqref{eqn:QE_and_dpeakon_transformation} as:
\[e^{x_{i+1}^{(n)}-x_i^{(n)}}=\frac{ \tau_{i+1}^{(-n,n)}\tau_{i-1}^{(-n+1,n)}\tau_i^{(-n,n+1)}\tau_{i-1}^{(-n+1,n+1)} }{ \tau_i^{(-n,n)}\tau_i^{(-n+1,n)}\tau_{i-1}^{(-n,n+1)}\tau_i^{(-n+1,n+1)} },\quad i=1,2,\ldots,N-1.\]
Taking the logarithm of the both sides gives:
\[
x_{i+1}^{(n)}-x_i^{(n)}=
\log\frac{ \tau_{i+1}^{(-n,n)}\tau_i^{(-n,n+1)} }{ \tau_i^{(-n+1,n)}\tau_i^{(-n+1,n+1)} }
-\log\frac{ \tau_i^{(-n,n)}\tau_{i-1}^{(-n,n+1)} }{ \tau_{i-1}^{(-n+1,n)}\tau_{i-1}^{(-n+1,n+1)} },\quad i=1,2,\ldots,N-1.
\]
We thus arrive at the second equation of \eqref{eqn:solution_dpeakon}.
\end{proof}

Taking the limit as $n\to\infty$ in \eqref{eqn:determinant_expand} with $s=-n+\bar{s}$, we obtain the asymptotic expansion as $n\to\infty$ of the Hankel determinants $\tau_i^{(-n+\bar{s},n)}$.
\begin{proposition}\label{prop:determinat_expansion}
Assume that nonzero constants $\lambda_1,\lambda_2,\ldots,\lambda_N$ satisfy $|1+\delta^{(n)}/\lambda_1|>|1+\delta^{(n)}/\lambda_2|>\cdots>|1+\delta^{(n)}/\lambda_N|$.  
Then, there are constants $\varrho_1,\varrho_2,\ldots,\varrho_{N-1}\in[0,1)$ satisfying $\varrho_i>|1+\delta^{(n)}/\lambda_{i+1}|/|1+\delta^{(n)}/\lambda_i|$, and it holds that, for sufficiently large $n$:
\begin{equation}\label{eqn:asymptotic_determinant_expand}
\begin{aligned}
&\tau_i^{(-n+\bar{s},n)}=\left\{\prod_{k=1}^i \left[c_k\lambda_k^{\bar{s}}\prod_{\ell=0}^{n-1}\left(1+\frac{\delta^{(\ell)}}{\lambda_k}\right)\right]\prod_{1\leq k<\ell\leq i}(\lambda_k-\lambda_\ell)^2\right\}\left( 1+O(\varrho_i^n) \right),\\
&\quad i=0,1,\ldots,N-1.
\end{aligned}
\end{equation}
\end{proposition}
\begin{proof}
Equation \eqref{eqn:determinant_expand} with $s=-n+\bar{s}$ immediately leads to:
\[
\tau_i^{(-n+\bar{s},n)}=\sum_{1\leq j_1<j_2<\cdots<j_i\leq N}
C_{j_1,j_2,\ldots,j_i}^{(\bar{s})}
\prod_{k=1}^i \prod_{\ell=0}^{n-1}\left(1+\frac{\delta^{(\ell)}}{\lambda_{j_k}}\right),\quad i=1,2,\ldots,N-1,
\]
where $C_{j_1,j_2,\ldots,j_i}^{(\bar{s})}$ are constants defined using $c_{j_1},c_{j_2},\ldots,c_{j_i}$ and $\lambda_{j_1},\lambda_{j_2},\ldots,\lambda_{j_i}$ according to:
\[
C_{j_1,j_2,\ldots,j_i}^{(\bar{s})}\coloneqq 
\prod_{k=1}^i c_{j_k}\lambda_{j_k}^{\bar{s}} \cdot \prod_{1\leq k<\ell\leq i}(\lambda_{j_k}-\lambda_{j_\ell})^2,\quad i=1,2,\ldots,N-1.
\]
Thus, for sufficiently large $n$, we can expand the Hankel determinants $\tau_i^{(-n+\bar{s},n)}$ as \eqref{eqn:asymptotic_determinant_expand}. 
\end{proof}

Combining Proposition \ref{prop:determinat_expansion} with Theorem \ref{thm:solution_dpeakon}, we observe that, as $n\to\infty$:
\begin{equation}\label{eqn:asymptotic_behaviors_dpeakon_proof1}
\left\{\begin{aligned}
&p_i^{(n)}=\frac{1}{\lambda_i},\quad i=1,2,\ldots,N,\\
&x_i^{(n)}=\log \left[ c_i\prod_{k=1}^{i-1}\left(1-\frac{\lambda_i}{\lambda_k}\right)^2
\cdot\prod_{\ell=0}^{n-1}\left(1+\frac{\delta^{(\ell)}}{\lambda_i}\right)
\frac{ 1+O(\varrho_i^{n}) }{ 1+O(\varrho_{i-1}^n) } \right],\quad i=1,2,\ldots,N.
\end{aligned}\right.
\end{equation}
We therefore arrive at asymptotic behavior as $n\to\infty$ in the discrete peakon equation \eqref{eqn:dpeakon}. 
\begin{theorem}\label{thm:asymptotic_behaviors_dpeakon}
Assume that $\lambda_1>0,\lambda_2>0,\ldots,\lambda_N>0$ and $|1+\delta^{(n)}/\lambda_1|>|1+\delta^{(n)}/\lambda_2|>\cdots>|1+\delta^{(n)}/\lambda_N|$. Moreover, assume that the parameter sequence $\{\delta^{(n)}\}$ satisfies $\delta^{(n)}>-\min_i\lambda_i$ and converges to some constant $\delta^\ast$ as $n\to\infty$. Then, it holds that:
\begin{equation}\label{eqn:asymptotic_behaviors_dpeakon}
\left\{\begin{aligned}
&\lim_{n\to\infty}p_i^{(n)}=\frac{1}{\lambda_i},\quad i=1,2,\ldots,N,\\
&\lim_{n\to\infty}x_i^{(n+1)}-x_i^{(n)}=\log\left(1+\frac{\delta^\ast}{\lambda_i}\right),\quad i=1,2,\ldots,N,\\
&\lim_{n\to\infty}x_{i+1}^{(n)}-x_i^{(n)}=-\infty,\quad i=1,2,\ldots,N-1.
\end{aligned}\right.
\end{equation}
\end{theorem}
\noindent
Since the values of $\delta^{(n)}$ can be arbitrarily set, they can be chosen to satisfy $\delta^{(n)}>-\min_i\lambda_i$. Thus, if $\lambda_1>0,\lambda_2>0,\ldots,\lambda_N>0$ in the initial setting, then we can draw peakon wave motion by using the discrete peakon equation \eqref{eqn:dpeakon}. 
By replacing continuous-time $t$ with discrete-time $n$ in \eqref{eqn:CH_solution}, we can rewrite the $N$-peakon solution to the CH equation \eqref{eqn:CH} as:
\begin{equation}\label{eqn:dCH_solution}
u^{(n)}(x) \coloneqq\sum_{i=1}^N p_i^{(n)} e^{-\vert x-x_i^{(n)}\vert}. 
\end{equation}
Thus, we can refer to \eqref{eqn:dpeakon} as the discrete CH peakon equation. Theorem \ref{thm:asymptotic_behaviors_dpeakon} implies that the amplitude and position of the $i$th peakon respectively converge to $1/\lambda_i$ and $\log(1+\delta^\ast/\lambda_i)$ as $n\to\infty$. We also recognize that the positions of peakons satisfy $x_1^{(n)}>x_2^{(n)}>\cdots>x_N^{(n)}$ as $n\to\infty$. 

We now present three numerical examples concerning peakon wave motions described by the discrete peakon equation \eqref{eqn:dpeakon}. We used a computer with a Mac OS Ventura (ver.~13.0.1 (22A400)) operating system and an Apple M1 CPU. We performed double-precision floating-point operations using the numerical computation software Python 3.9.7.

\begin{figure}[tb]
\includegraphics[width=0.6\textwidth]{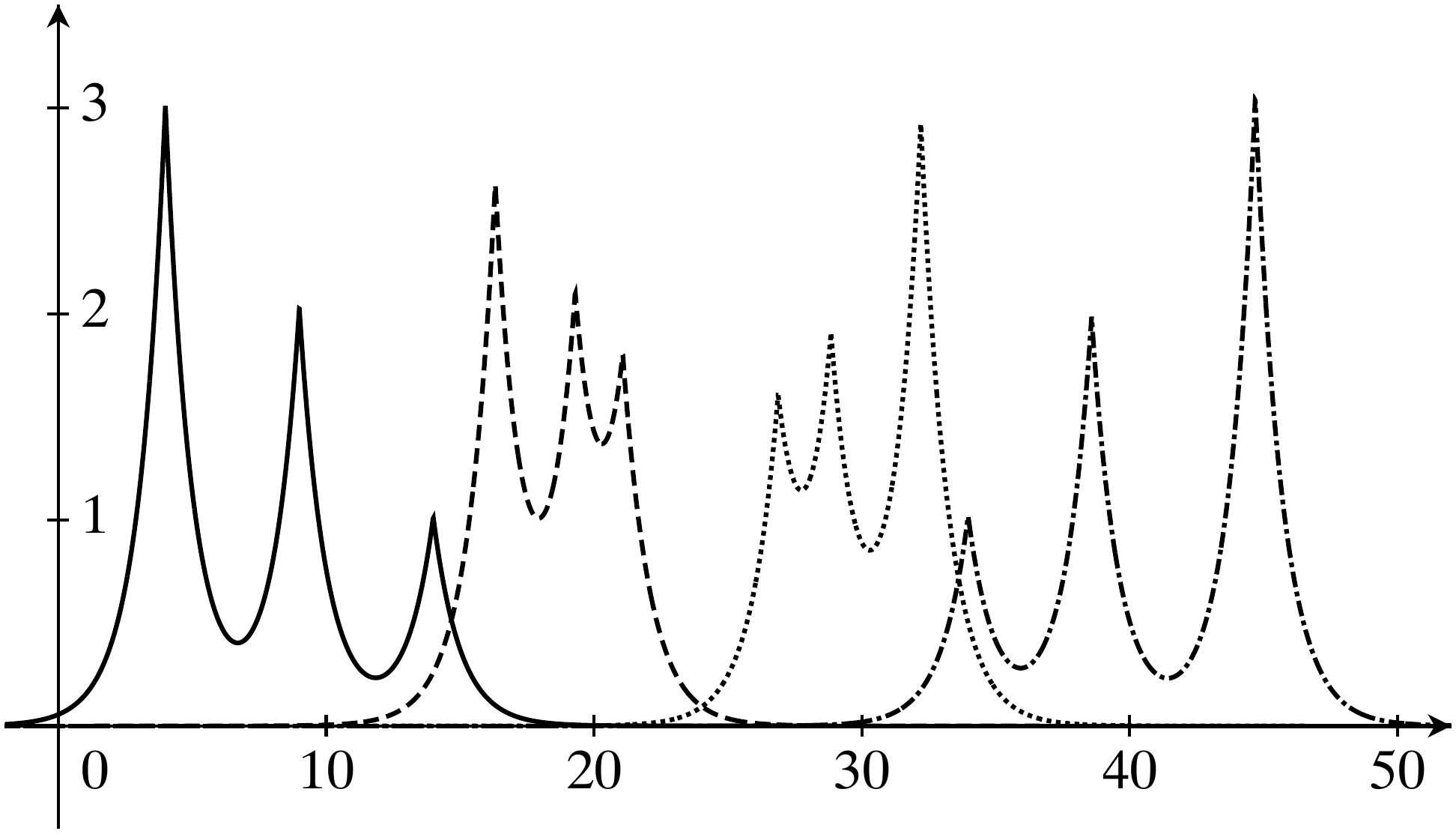}
\centering
\caption{2D motion of three peakon waves. Peakon position ($x$ axis) versus function values of $u^{(0)}(x),u^{(9)}(x),u^{(18)}(x)$ and $u^{(27)}(x)$ ($y$ axis). Solid line: $u^{(0)}(x)$; dashed line: $u^{(9)}(x)$; dotted line: $u^{(18)}(x)$ and dash-dotted line: $u^{(27)}(x)$.} \label{fig2:3peakons}  
\end{figure}
\begin{figure}[tb]
\centering
\includegraphics[width=0.6\textwidth]{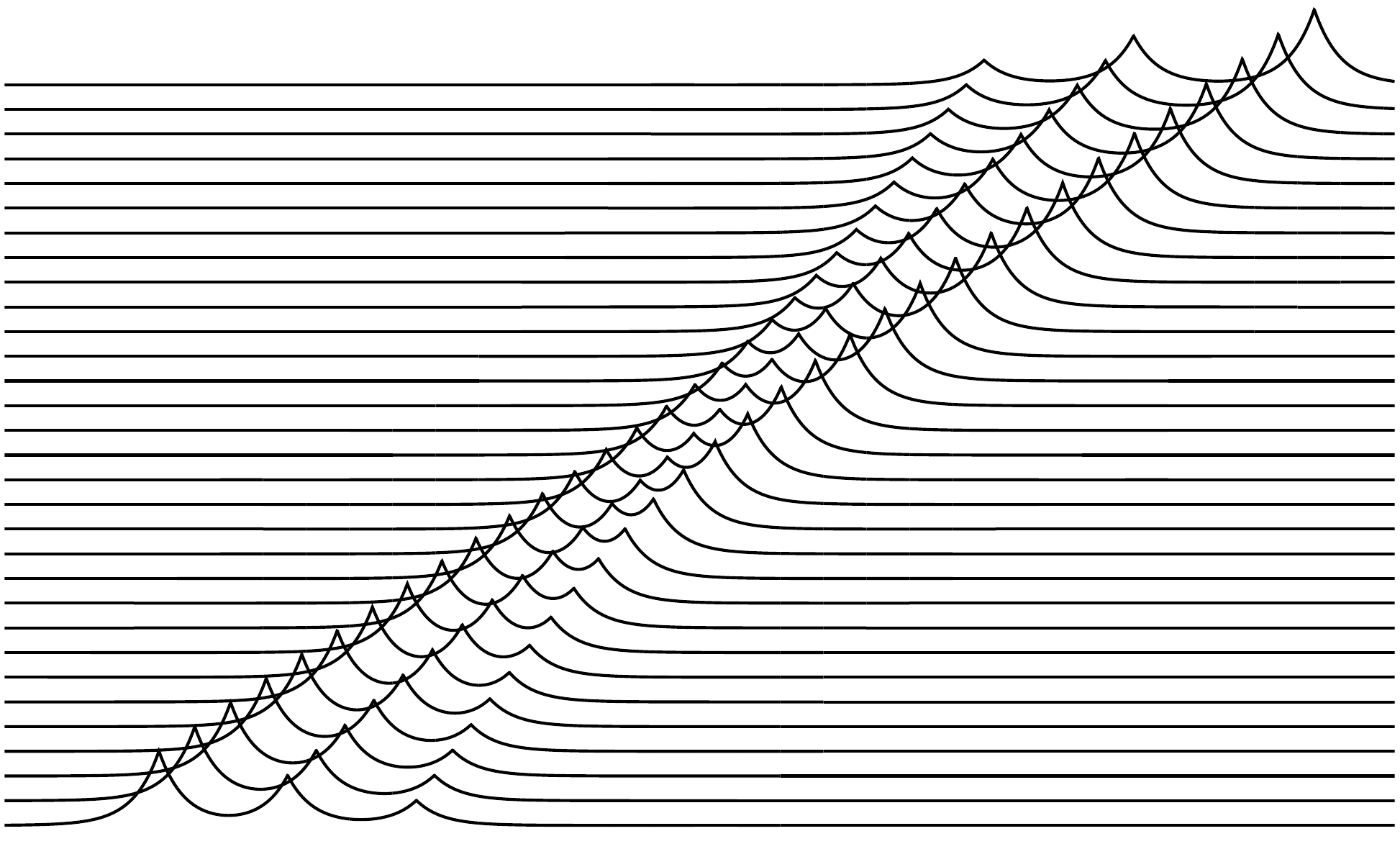}
\caption{3D plot of the first example.} \label{fig3:3peakons}
\end{figure}
We first handle the case where $p_1^{(0)}=1,p_2^{(0)}=2,p_3^{(0)}=3,x_1^{(0)}=14,x_2^{(0)}=9,x_3^{(0)}=4$ and $N=3$. We fix the parameters $\delta^{(n)}$ to $\delta^{(n)}=1$ at any discrete-time. 
In this setting, by using the discrete peakon equation \eqref{eqn:dpeakon}, we generated the sequences $\{p_i^{(n)}\}$ and $\{x_i^{(n)}\}$. 
Figure~\ref{fig2:3peakons} displays the peakon waves at discrete-time $n=0,9,18,27$, namely, four functions of $u^{(0)}(x),u^{(9)}(x),u^{(18)}(x)$ and $u^{(27)}(x)$. 
Figure~\ref{fig3:3peakons} displays a 3D plot corresponding to Fig.~\ref{fig2:3peakons}.
From Fig.~\ref{fig2:3peakons} and Fig.~\ref{fig3:3peakons}, we can see that leftmost peakon with the highest amplitude at $n=0$ moves to the rightmost under discrete-time evolution from $n=0$ to $n=27$. 
We also compute that, for $n=136,137,\dots,$ $p_1^{(n)}=p_1^{(135)}=3.0530128591497787,$ $p_2^{(n)}=p_2^{(135)}=1.9805447136469283$ and $p_3^{(n)}=p_3^{(135)}=0.98672436809518$. 
Thus, by using the first equation of \eqref{eqn:asymptotic_behaviors_dpeakon}
, we can estimate the values of $1/\lambda_1,1/\lambda_2$ and $1/\lambda_3$ as 
$3.053012859149779,1.980544713646928$ and $0.986724368095180$, respectively. 
The values of $\vert (x_1^{(136)}-x_1^{(135)}-\log(1+\delta^\ast/\lambda_1))/\log(1+\delta^\ast/\lambda_1)\vert,\vert (x_2^{(136)}-x_2^{(135)}-\log(1+\delta^\ast/\lambda_2))/\log(1+\delta^\ast/\lambda_2)\vert$ and $\vert (x_3^{(136)}-x_3^{(135)}-\log(1+\delta^\ast/\lambda_3))/\log(1+\delta^\ast/\lambda_3)\vert$ are 
$3.173288587851057\times10^{-16},$ $2.033178006120378\times10^{-16}$ and $1.617252245736617\times10^{-16}$, respectively. 
In other words, for $i=1,2,3$, the value of $x_i^{(n+1)}-x_i^{(n)}$ approaches that of $\log(1+\delta^\ast/\lambda_i)$ at sufficiently large $n$, which implies the second equation of \eqref{eqn:asymptotic_behaviors_dpeakon}. 
We can verify that the peakon motion in the case where $N\neq3$ is similar to that in the case where $N=3$. 

The second example concerns an interaction of two peakons with positive and negative amplitude. The peakon with negative amplitude is often called an antipeakon. 
We consider the case where $p_1^{(0)}=-2,p_2^{(0)}=5,x_1^{(0)}=2+\log10,x_2^{(0)}=2$, and $N=2$, and set the discretization parameters $\delta^{(n)}=0.1$ at any discrete-time $n$. 
Discrete-time evolutions from $n=0$ to $n=6$ in the discrete peakon equation \eqref{eqn:dpeakon} give $p_1^{(6)}=-23.87217732451234,p_2^{(6)}=-18.65040845375860,e^{x_1^{(6)}}=60.73038625216492$ and $e^{x_2^{(6)}}=-85.90195612627288$. 
Remarkably, there does not exist an $x_2^{(6)}$ that satisfies $e^{x_2^{(6)}}=-85.90195612627288$. The discrete peakon equation \eqref{eqn:dpeakon}, however, generates discrete-time evolutions even after discrete-time $n=6$. 
This is because the discrete-time evolutions in \eqref{eqn:dpeakon} require the value of $e^{x_2^{(n)}}$ but not $x_i^{(n)}$. 
Figure~\ref{fig4:2peakons} and Figure~\ref{fig5:2peakons} depict the motions of a peakon and an antipeakon using the discrete peakon equation \eqref{eqn:dpeakon}. From Figs~\ref{fig4:2peakons} and \ref{fig5:2peakons}, we can observe that, as discrete-time passes, the peakon and antipeakon move closer, interact, swap position with each other, and then move apart. 
\begin{figure}[tb]
\centering
\includegraphics[width=0.6\textwidth]{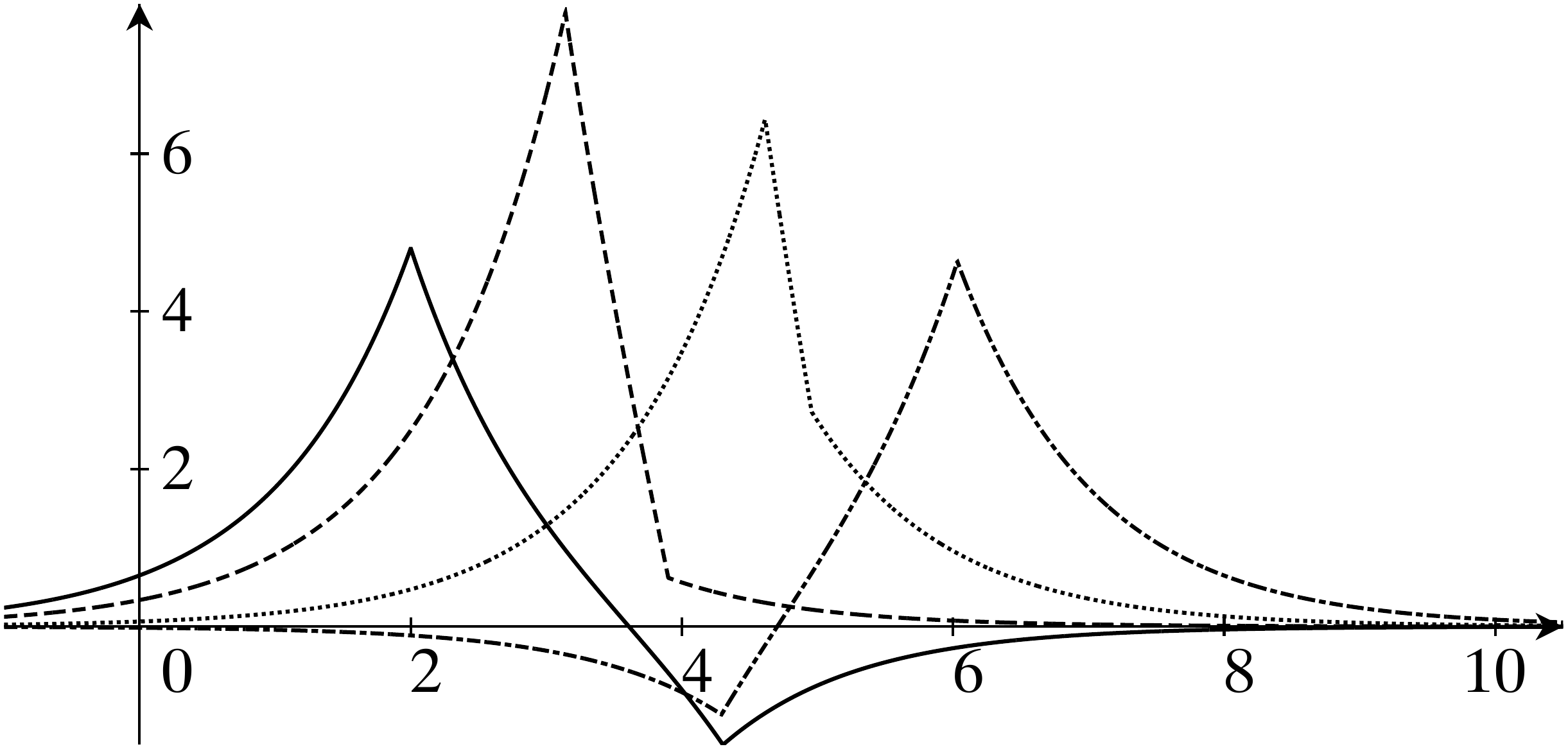}
\caption{2D motion of a peakon wave and an antipeakon wave showing peakon position ($x$ axis) versus function values of $u^{(0)}(x),u^{(4)}(x),u^{(8)}(x)$ and $u^{(12)}(x)$ ($y$ axis). Solid line: $u^{(0)}(x)$; dashed line: $u^{(4)}(x)$; dotted line: $u^{(8)}(x)$ and dash-dotted line: $u^{(12)}(x)$.} \label{fig4:2peakons}  
\end{figure}
\begin{figure}[tb]
\centering
\includegraphics[width=0.6\textwidth]{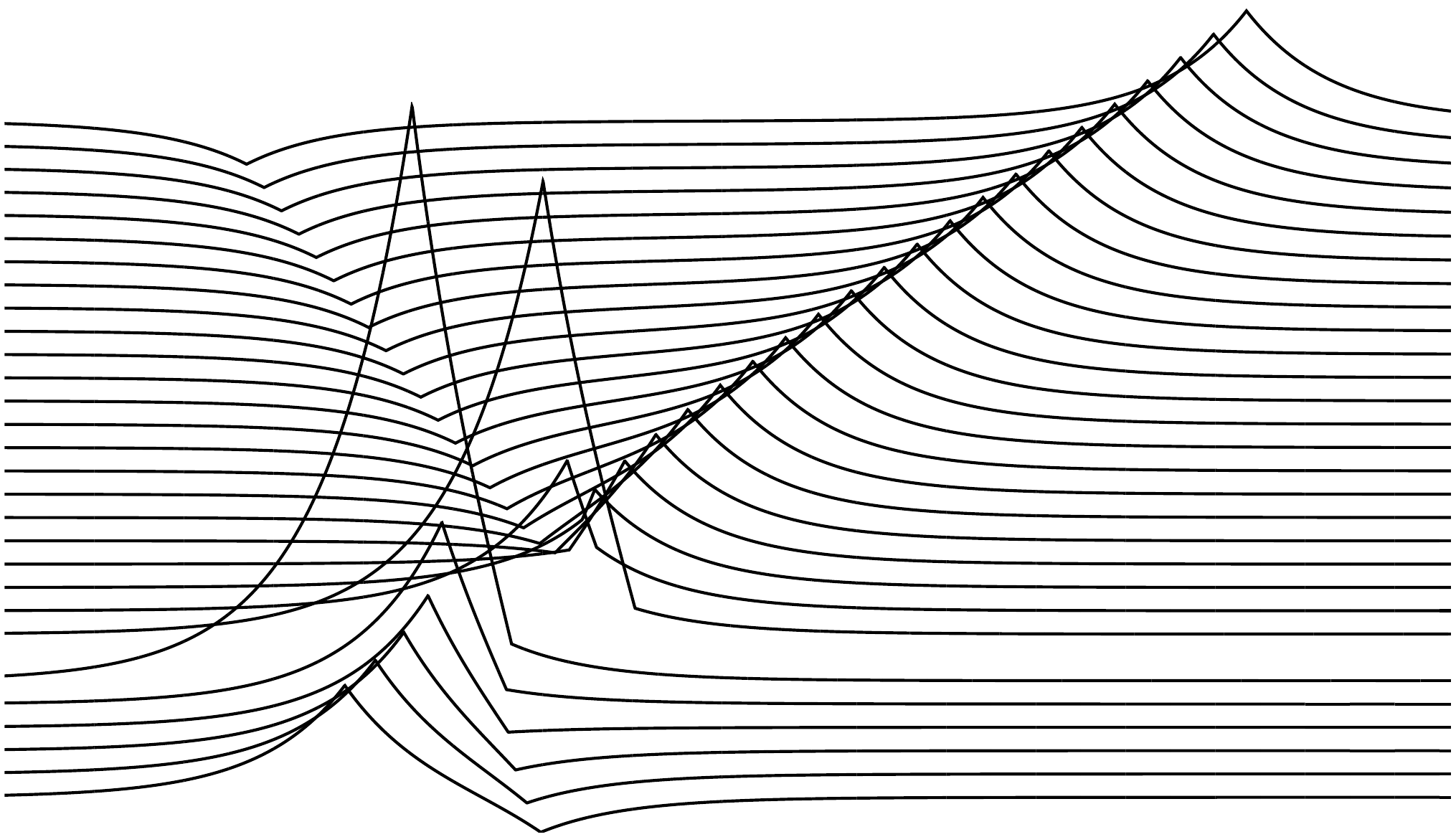}
\caption{3D plot of the second example.} \label{fig5:2peakons} 
\end{figure}

In the third example, we closely observe the motions of a peakon and an antipeakon approaching each other until just before they collide. 
Once again, we consider the case where $p_1^{(0)}=-2,p_2^{(0)}=5,x_1^{(0)}=2+\log10,x_2^{(0)}=2$ and $N=2$. We set the discretization parameters $\delta^{(n)}=0.1$ for $n=0,1,\ldots,4$ and to $\delta^{(n)}=0.02$ for $n=5,6,\ldots.$ 
It is worth noting here that we can change the parameters when necessary. 
Obviously, the discrete-time evolutions from $n=0$ to $n=4$ in \eqref{eqn:dpeakon} generate the same sequences as in the second example. We also obtain $p_1^{(11)}=-76.14039569735561,p_2^{(11)}=-150.8018405883023,e^{x_1^{(11)}}=67.03500835834600$ and $e^{x_2^{(11)}}=-35.073140188469985$. 
Thus, wave collision occurs at discrete time $n=11$, not at discrete time $n=6$. 
Figure~\ref{fig6:2peakons} presents the motions of the peakon and antipeakon at discrete-time $n=0,1,\ldots,10$. 
From Fig.~\ref{fig6:2peakons}, we see that as the two peakon waves near each other, the wave shape between a peakon and an antipeakon becomes perpendicular to the horizontal. 
\begin{figure}[tb]
\centering
\includegraphics[width=0.6\textwidth]{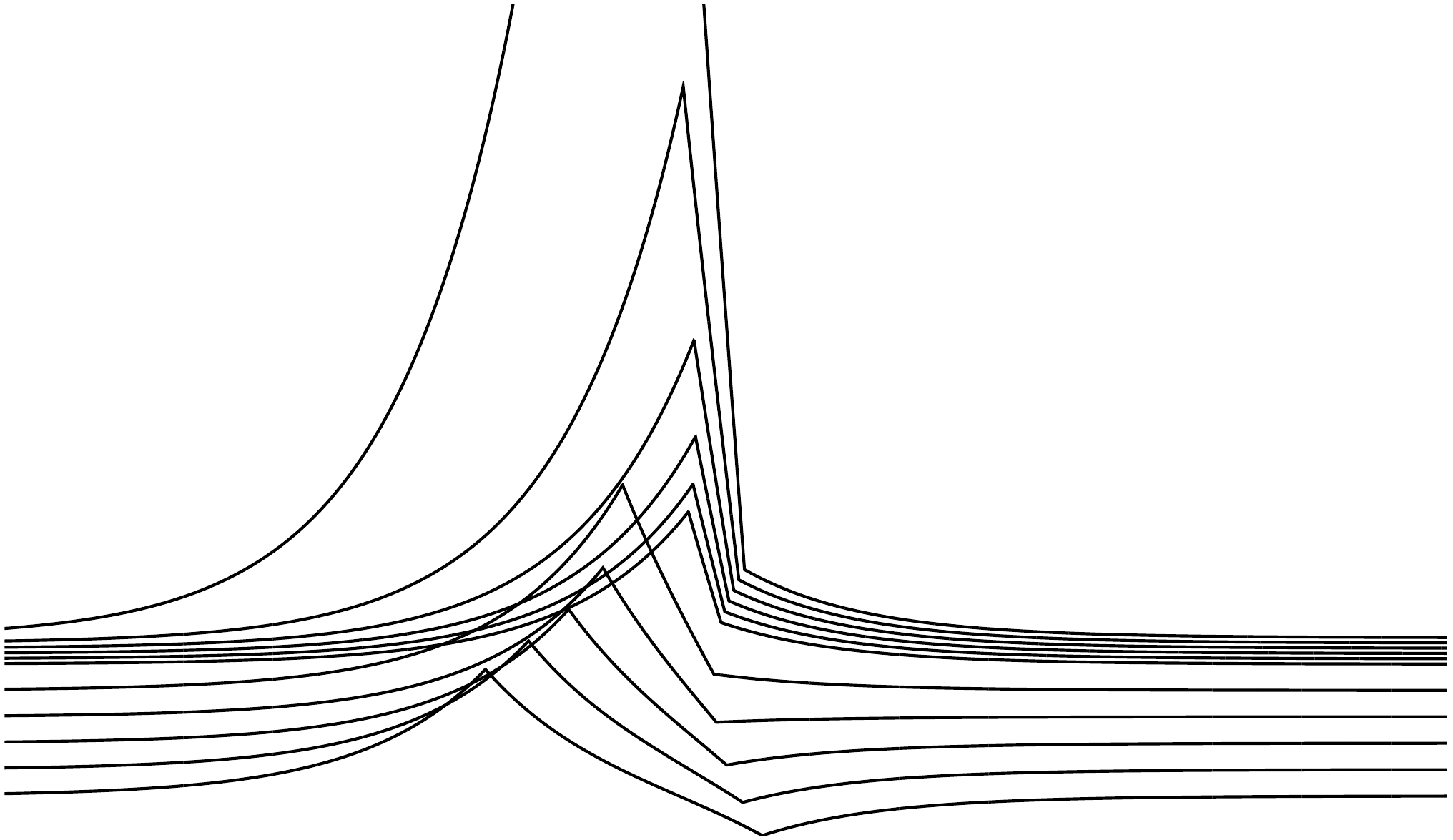}
\caption{3D plot of the third example.} \label{fig6:2peakons} 
\end{figure}
%

\section{Concluding remarks}
\label{sec:6}
%
In this paper, we first related the Camassa-Holm (CH) peakon equation to the Toda equation in which discrete-time evolutions generate similarity transformations of tridiagonal matrices.  Next, by focusing on spectral transformations of orthogonal polynomials, we examined some properties concerning discrete-time evolutions of orthogonal polynomials. From the matrix used in the discrete-time evolutions of the orthogonal polynomials, we derived shifted $LR$-like transformations, and subsequently arrived at a time-discretization version of the CH peakon equation. 
To validate this discretization, we showed that a continuous-limit of the shifted $LR$-like transformation coincides with the Lax representation of the CH peakon equation. 
It is important to note that we successfully achieved a nonautonomous discretization of the CH peakon equation through spectral transformations of orthogonal polynomials.
Finally, we expressed the solution to the discrete CH peakon equation using Hankel determinants, and then clarified asymptotic behaviors in the discrete CH peakon equation. We also presented three examples to numerically demonstrate peakon wave motion. 

As shown in Section \ref{sec:4}, the perspective of matrix $LR$ transformations is useful in relating the discrete CH peakon equation to the peakon lattice equation~\cite{Ragnisco_1996}. 
An aim of future work is thus to apply the $LR$ perspective to clarify relationships between the CH equation, which is the origin of the discrete CH peakon equation, and its derivatives such as the semi-discrete CH equation~\cite{Feng_2010,Ohta_2008}, the full-discrete CH equation and the CH peakon equation. 
Another future work is to design a numerical algorithm for computing matrix eigenvalues based on the discrete CH peakon equation. 
\section*{Acknowledgements}
This work was partially supported by the joint project of Kyoto University and Toyota Motor
Corporation, titled ``Advanced Mathematical Science for Mobility Society". The research of ST was supported by JSPS Grant-in-Aid for Scientific Research (B), 19H01792. 
%
%

%
%
%

\begin{thebibliography}{70}
%
\bibitem{Beals_2000}Beals R, Sattinger D H and Szmigielski J 2000
Multipeakons and the classical moment problem 
{\it Adv. Math.} {\bf 154} 229--57
%
\bibitem{Beals_2001}Beals R, Sattinger D H and Szmigielski J 2001
Peakons, strings, and the finite Toda lattice
{\it Commun. Pure Appl. Math.} {\bf 54} 91--106
%
\bibitem{Chang_2018_BKP}Chang X, Hu X, Li S and Zhao J 2018
An application of Pfaffians to multipeakons of the Novikov equation and the finite Toda lattice of BKP type
{\it Adv. Math.} {\bf 338} 1077--118
%
\bibitem{Chang_2018_CKP}Chang X, Hu X and Li S 2018
Degasperis-Procesi peakon dynamical system and finite Toda lattice of CKP type 
{\it Nonlinearity} 
{\bf 31} 4746--75
%
\bibitem{Chang_2020}Chang X, Hu X, Szmigielski J and Zhedanov A 2020
Isospectral flows related to Frobenius-Stickelberger-Thiele polynomials 
{\it Commun. Math. Phys.} {\bf 377} 387--419
%
\bibitem{Camassa_1993}Camassa R and Holm D D 1993
An integrable shallow water equation with peaked solitons 
{\it Phys. Rev. Lett.} 
{\bf 71} 1661--4
%
\bibitem{Camassa_1994}Camassa R, Holm D D and Hyman J 1994
A new integrable shallow water equation 
{\it Adv. Appl. Mech.}  {\bf 31} 1--33
%
\bibitem{Chihara_1978}Chihara T S 1978
{\it An introduction to orthogonal polynomials} 
(New York-London-Paris: Gordon and Breach Science Publishers)
%
\bibitem{Flaschka_1974_I}Flaschka H 1974
On the Toda lattice I, Existence of integrals 
{\it Phys. Rev. B} {\bf 9} 1924--5
%
\bibitem{Flaschka_1974_II}Flaschka H 1974
On the Toda lattice II, Inverse-scattering solution 
{\it Prog. Theor. Phys.} {\bf 51} 703--16
%
\bibitem{Feng_2010}Feng B F, Maruno K and Ohta Y 2010
A self-adaptive moving mesh method for the Camassa-Holm equation
{\it J. Comput. Appl. Math.} {\bf 235} 229--43
%
\bibitem{Fukuda_2013a}Fukuda A, Ishiwata E, Yamamoto Y, Iwasaki M and Nakamura Y 2013 
Integrable discrete hungry systems and their related matrix eigenvalues 
{\it Annal. Mat. Pura Appl.} {\bf 192} 423--45
%
\bibitem{Fukuda_2013b}Fukuda A, Yamamoto Y, Iwasaki M, Ishiwata E and Nakamura Y 2013 
On a shifted $LR$ transformation derived from the discrete hungry Toda equation 
{\it Monatsh. Math.} {\bf 170} 11--26
%
\bibitem{Hirota_1993}Hirota R, Tsujimoto S and Imai T 1993
Difference scheme of soliton equation ({\it Future  Directions  of  Nonlinear  Dynamics in  Physical  and  Biological  Systems} vol 312 pp.7--15) ed Christiansen P L, Eilbeck J C and Parmentier R D (New York: Plenum)
%
%
%
\bibitem{Holden_2006}Holden H and Raynaud X 2006
A convergent numerical scheme for the Camassa-Holm equation based on multipeakons 
{\it Discrete Contin. Dyn. Syst.} {\bf 3} 505--23
%
\bibitem{Iwasaki_2002}Iwasaki M and Nakamura Y 2002
On the convergence of a solution of the discrete Lotka-Volterra system 
{\it Inverse Probl.} 
{\bf 18} 1569--78
%
\bibitem{Kac_1975}Kac M and Van Moerbekc P 1975
On an explicitly soluble system of nonlinear differential equations related to certain Toda lattices 
{\it Adv. in Math.} {\bf 16} 160--9
%
\bibitem{Kobayashi_2021}Kobayashi K 2021
Nonautonomous discrete elementary Toda orbits and their ultradiscretization
{\it J. Phys. A} 
{\bf 54} 455203 (19pp)
%
\bibitem{Maeda_2013}Maeda K and Tsujimoto S 2013
Direct connection between the RII chain and the nonautonomous discrete modified KdV lattice 
{\it SIGMA} {\bf 9} 73--86
%
\bibitem{Manakov_1975}Manakov S V 1974
Complete integrability and stochastization of discrete dynamical systems 
{\it Sov. Phys. JETP} {\bf 40} 269--74
%
\bibitem{Moser_1975}Moser J 1975
Finitely many mass points on the line under the influence of an exponential potential -- An integrable systems ({\it Dynamical Systems, Theory and Applications} {\it Lect. Notes Phys.} vol 38) ed Moser J (Berlin: Springer-Verlag)
%
\bibitem{Ohta_2008}Ohta Y, Maruno K and Feng B F 2008
An integrable semi-discretization of the Camassa--Holm equation and its determinant solution 
{\it J. Phys. A} 
{\bf 41} 355205
%
\bibitem{Ragnisco_1996}Ragnisco O and Bruschi M 1996
Peakons, $r$-matrix and Toda lattice 
{\it Physica A} {\bf 228} 150--9
%
\bibitem{Spiridonov_1995}Spiridonov V and Zhedanov A 1995
Discrete Darboux transformations, the discrete-time Toda lattice, and the Askey-Wilson polynomials 
{\it Methods and Appl. Anal.} {\bf 2} 369--98
%
\bibitem{Spiridonov_1997}Spiridonov V and Zhedanov A 1997
Discrete-time Volterra chain and classical orthogonal polynomials, 
{\it J. Phys. A} 
{\bf 30} 8727--37
%
\bibitem{Suris_1996a}Suris Y B 1996
Integrable discretizations of the Bogoyavlensky lattices
{\it J. Math. Phys.} 
{\bf 37} 3982--96
%
\bibitem{Suris_1996b}Suris Y B 1996
A discrete-time relativistic Toda lattice 
{\it J. Phys. A} 
{\bf 29} 451--65
%
\bibitem{Symes_1982}Symes W W 1982
The QR algorithm and scattering for the finite nonperiodic Toda lattice
{\it Physica D} {\bf 4} 275--80
%
\bibitem{Toda_1967}Toda M 1967
Vibration of a chain with nonlinear integration 
{\it J. Phys. Soc. Jpn.} 
{\bf 22} 431--6
%
\bibitem{Tokihiro_1999}Tokihiro T, Nagai A and Satsuma J 1999
Proof of solitonical nature of box and ball systems by means of inverse ultra-discretization 
{\it Inverse Probl.} 
{\bf 15} 1639--1662
%
\bibitem{Tsujimoto_1993}Tsujimoto S, Hirota R and Oishi S 1993
An extension and discretization of Volterra equation I 
{\it Tech. Rep. Proc. IEICE NLP} {\bf 92} 1--3
%
\bibitem{Tsujimoto_2002}Tsujimoto S, Nakamura Y and Iwasaki M 2002
The discrete Lotka-Volterra system computes singular values 
{\it Inverse Prob.} 
{\bf 17} 53--8
%
\bibitem{Vandebril_2008}Vandebril R, Barel M V and Mastronardi N 2008
{\it Matrix computations and semiseparable matrices}
vol 1 Linear Systems (Baltimore, MD: Johns Hopkins University Press)
%
\bibitem{Yamamoto_2022}Yamamoto Y, Minoshita N and Iwasaki M 2020
Discrete relativistic Toda equation from the perspective of shifted $LR$ transformation
{\it Physica D} {\bf 440} 133485
%
\bibitem{Zhedanov_1997}Zhedanov A 1997 
Rational spectral transformations and orthogonal polynomials
{\it J. Comput. Appl. Math.} {\bf 85} 67--86
%
\end{thebibliography}
\end{document}